\documentclass[conference]{IEEEtran}
\IEEEoverridecommandlockouts
 \usepackage{cite}
\usepackage{algpseudocode}
\usepackage{amsfonts}
\usepackage[hidelinks]{hyperref}
\usepackage{booktabs} 
\usepackage{enumerate}
\usepackage{soul}
\usepackage{framed}
\usepackage{paralist}
\usepackage[shortlabels]{enumitem}
\usepackage{subfigure}
\usepackage[linesnumbered,lined,boxed]{algorithm2e}
\newcommand{\Adj}{\text{Adj}}

\newcommand{\D}{\mathsf D}
\newcommand{\R}{\mathsf R}

\newcommand{\Prob}{\mathbb P}
\newcommand{\RN}[1]{%
  \textup{\uppercase\expandafter{\romannumeral#1}}%
}

\usepackage{booktabs}
\usepackage{bigints}
\usepackage{hyperref}

\usepackage{adjustbox,multirow}
\newcommand{\subhead}[1]{\vspace {0.04in}\noindent{\textbf{#1.}}}

\newcommand{\DP}{differential privacy }
\newcommand\norm[1]{\left\lVert#1\right\rVert}
\usepackage{mathtools}
\usepackage[thinc]{esdiff}
\usepackage{courier}
\usepackage{amsthm}

\newcommand\ceil[1]{\lceil#1\rceil}
\newtheorem{defn}{Definition}[section]
\usepackage{bbm, dsfont}
\newtheorem{thm}{Theorem}[section]
\newtheorem{lem}[thm]{Lemma}

\usepackage{bbm, dsfont}
\usepackage{tikz}
\usepackage{amsmath}

\def\BibTeX{{\rm B\kern-.05em{\sc i\kern-.025em b}\kern-.08em
    T\kern-.1667em\lower.7ex\hbox{E}\kern-.125emX}}
\begin{document}

\title{\title{\huge {\it DPOAD}: {\it D}ifferentially {\it P}rivate {\it O}utsourcing of {\it A}nomaly {\it D}etection through Iterative Sensitivity Learning}
\thanks{Identify applicable funding agency here. If none, delete this.}
}
\author{\IEEEauthorblockN{Meisam Mohammady\IEEEauthorrefmark{1}, Han Wang \IEEEauthorrefmark{2}, Lingyu Wang \IEEEauthorrefmark{3}, Mengyuan Zhang \IEEEauthorrefmark{4}}
\IEEEauthorblockN{Yosr Jarraya \IEEEauthorrefmark{5}, Suryadipta Majumdar \IEEEauthorrefmark{3}, Makan Pourzandi \IEEEauthorrefmark{5}, Mourad Debbabi \IEEEauthorrefmark{3}, Yuan Hong  \IEEEauthorrefmark{6}}
\IEEEauthorblockA{ \IEEEauthorrefmark{1} Data61, CSIRO, Cyber Security CRC\\
\IEEEauthorrefmark{2} Illinois Institute of Technology, \IEEEauthorrefmark{3} Concordia University\\ \IEEEauthorrefmark{4} Hong Kong Polytechnic University, \IEEEauthorrefmark{5} Ericsson Security Research, Ericsson Canada, \IEEEauthorrefmark{6} University of Connecticut\\
meisam.mohammady@csiro.au, hwang185@hawk.iit.edu, yuan.hong@uconn.edu, \{wang, suryadipta.majumdar, debbabi\}@concordia.ca} mengyuan.zhang@polyu.edu.hk, \{yosr.jarraya, makan.pourzandi\}@ericsson.com
}
\maketitle

\begin{abstract}
Outsourcing anomaly detection to third-parties can allow data owners to overcome resource constraints (e.g., in lightweight IoT devices), facilitate collaborative analysis (e.g., under distributed or multi-party scenarios), and benefit from lower costs and specialized expertise (e.g., of Managed Security Service Providers). Despite such benefits, a data owner may feel reluctant to outsource anomaly detection without sufficient privacy protection. To that end, most existing privacy solutions would face a novel challenge, i.e., preserving privacy usually requires the difference between data entries to be eliminated or reduced, whereas anomaly detection critically depends on that difference. Such a conflict is recently resolved under a local analysis setting with trusted analysts (where no outsourcing is involved) through moving the focus of differential privacy (DP) guarantee from ``all'' to only ``benign'' entries. In this paper, we observe that such an approach is not directly applicable to the outsourcing setting, because data owners do not know which entries are ``benign'' prior to outsourcing, and hence cannot selectively apply DP on data entries. Therefore, we propose a novel \emph{iterative} solution for the data owner to gradually ``disentangle'' the anomalous entries from the benign ones such that the third-party analyst can produce accurate anomaly results with sufficient DP guarantee. We design and implement our {\it D}ifferentially {\it P}rivate {\it O}utsourcing of {\it A}nomaly {\it D}etection ({\it DPOAD}) framework, and demonstrate its benefits over baseline Laplace and PainFree mechanisms through experiments with real data from different application domains. 
\end{abstract}

\begin{IEEEkeywords}
component, formatting, style, styling, insert
\end{IEEEkeywords}

\maketitle
\thispagestyle{empty}
 \section{Introduction}
\label{intro4}

The recent worldwide surge in the adoption of cloud/edge computing necessitated by the lockdown and remote working would likely sustain even post-pandemic due to the apparent benefits of outsourcing computing tasks to clouds, such as the pervasive access, increased flexibility and efficiency, and reduced costs~\footnote{For example, Microsoft said it had seen ``two years of digital transformation in two months'' as its customers started adopting cloud solutions, and Gartner projected that the worldwide spending on public cloud services would grow by 18.4\% in 2021 to a total of \$304.9 billion.}. In such a context, there is also an increasing need for outsourcing security analyses, such as anomaly detection, to third parties hosted in clouds, e.g., Managed Security Service Providers (MSSPs)~\cite{10.1145/1107622.1107648}. Outsourcing security analyses to clouds offers several benefits. First, it can allow light-weight distributed systems (especially in IoT) to overcome their resource constraints (e.g., limited computational power, storage, or battery of lightweight smart devices). Second, it can enable collaborative analysis that involves multi-party data sharing (e.g., detection of a common malware infection across different smart homes). Third, it can bring various benefits of hiring an MSSP, such as cost effectiveness, updated technologies, and better trained and specialized experts with knowledge about a diverse range of clients and security vulnerabilities. 

However, most data owners would be reluctant to  outsource their security analyses without sufficient privacy protection due to the sensitive nature of such analyses~\cite{risk}. 
In particular, outsourcing anomaly detection raises an interesting and unique challenge to existing privacy solutions. On one hand, most outsourcing and collaborative approaches (e.g.,~\cite{nguyen2019diot}) do not allow outsourcing anomaly detection, and may sacrifice privacy guarantee for accuracy . On the other hand, existing privacy measures, such as differential privacy (DP) (which has been widely recognized as the state-of-the-art privacy notion~\cite{10.1007/978-3-540-79228-4_1,10.1007/11681878_14}), often require the differences between data objects to be eliminated or reduced (e.g., DP aims to generate indistinguishable outputs), whereas anomaly detection critically relies on such differences to identify outliers. 
In particular, Hafiz et al.~\cite{DBLP:conf/cscml/BittnerSW18,DBLP:conf/ccs/AsifPV19}) show that relatively strong privacy guarantees render the analysis results useless (one false positive in each three detected anomalies). 

Such a ``conflict'' between anomaly detection and privacy is recently resolved under a \emph{local analysis setting} where no outsourcing is involved and the analyst is trusted to know the distribution of data. The key idea behind those solutions, namely the ``anomaly-restricted'' strategy, is to only preserve the privacy for benign records (instead of all records, as in a normal application of DP). Unfortunately, such a strategy is no longer applicable when it comes to the \emph{outsourced analysis setting} where the (third-party) analyst is not trusted to know the data distribution, whereas the data owner does not know which entries are ``benign”, and hence cannot selectively apply DP on data entries prior to outsourcing. In particular, Figure~\ref{fig:1} compares the system models of the \emph{local analysis setting}, as addressed in existing works~\cite{DBLP:conf/cscml/BittnerSW18, bohler2017privacy,DBLP:conf/pkdd/OkadaFS15,DBLP:conf/ccs/AsifPV19, DBLP:conf/coinco/AsifTVSA16} (top), and the \emph{outsourced analysis setting} studied in this work. 

\begin{itemize}
\item In the first model (\emph{local analysis setting}), the data owner (who is also the analyst) is a trusted data center (e.g., a Bank) who both applies DP to the data and performs anomaly detection. On the other hand, the external user is untrusted (for seeing the original data), and can only access DP-protected detection results. Under such a model, we can see that the trusted data center can easily apply the anomaly-restricted solutions~\cite{DBLP:conf/cscml/BittnerSW18,DBLP:conf/ccs/AsifPV19} by selectively applying DP on only the benign records, as it has access to the original data.

\item The second model (\emph{outsourced analysis setting}) is similar to the first model, except that the data owner and analyst are now decoupled due to outsourcing, e.g., the data owner is a lightweight IoT device who is trusted but incapable of performing anomaly detection by itself, whereas the analyst is an untrusted (for seeing the original data) cloud service provider to whom the detection has been outsourced. 
Under such a model, we can see that the data owner can no longer apply the anomaly-restricted solutions~\cite{DBLP:conf/cscml/BittnerSW18,DBLP:conf/ccs/AsifPV19}, as it does not know which records are benign before sending the data to the analyst. On the other hand, the data owner cannot send the data without privacy protection (due to untrusted analyst), or simply apply DP on all records (which could render the records indistinguishable to anomaly detection). 
\end{itemize}

Therefore, the outsourced analysis setting creates a unique challenge to privacy protection. Specifically, we need to resolve the cyclic dependencies between the following research questions:
\begin{inparaenum}[(i)]
\item How can the analyst provide the data owner with enough information, e.g., an approximate likelihood of outlierness, to facilitate the anomaly-restricted DP solutions~\cite{DBLP:conf/cscml/BittnerSW18,DBLP:conf/ccs/AsifPV19}.   
\item How can the data owner provide an (honest but untrusted) analyst with a DP-protected version of its data that can lead to sufficiently accurate anomaly detection.
\end{inparaenum}
\begin{figure}[!h]
	\centering		\includegraphics[angle=0, width=1\linewidth]{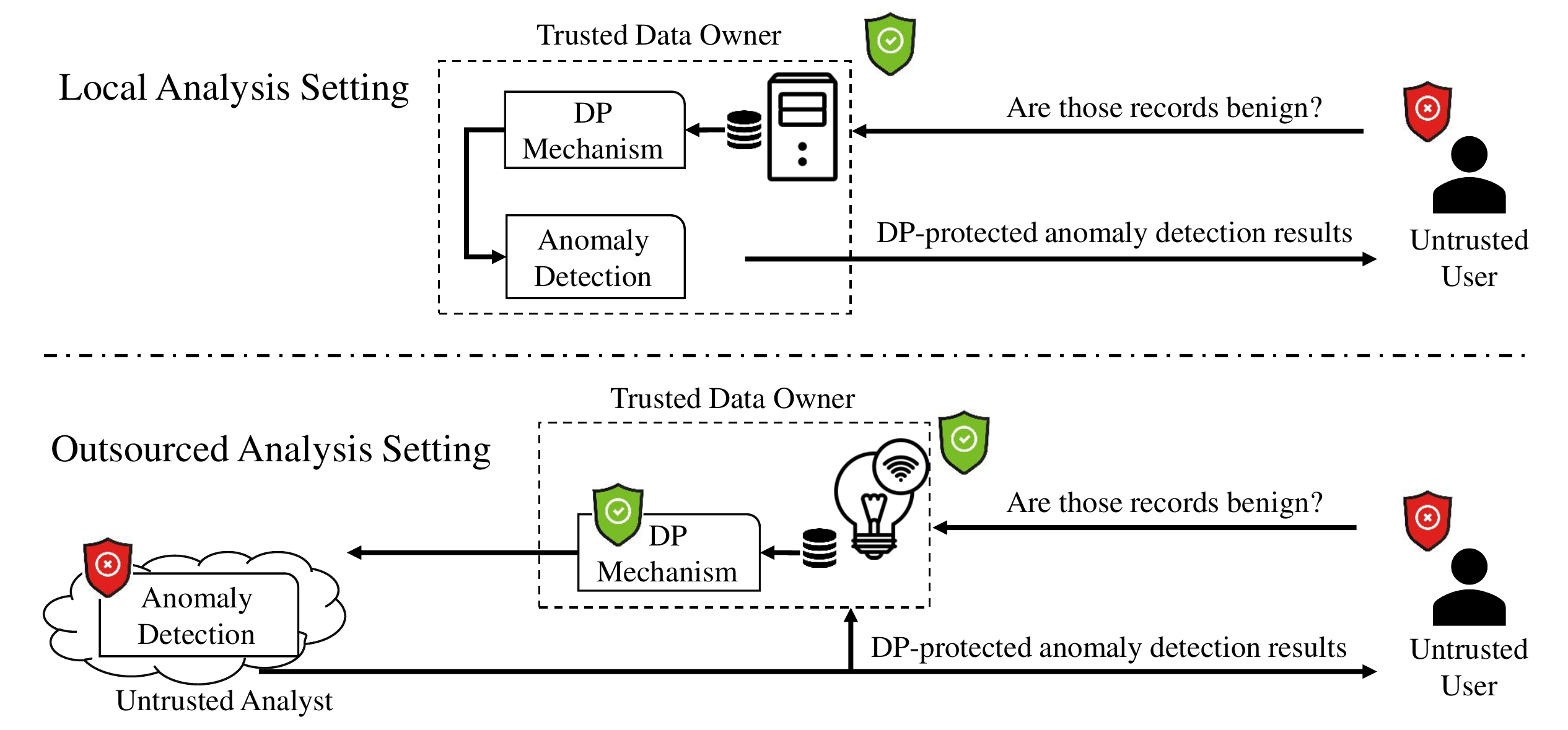}
\caption{System models of existing works~\cite{DBLP:conf/cscml/BittnerSW18, bohler2017privacy,DBLP:conf/pkdd/OkadaFS15,DBLP:conf/ccs/AsifPV19, DBLP:conf/coinco/AsifTVSA16} (top), and of DPOAD (down)}
	\label{fig:1}
\end{figure}

\subhead{Main Idea} 
To address the aforementioned challenges, our main idea is an iterative approach for the data owner and analyst to gradually ``disentangle'' the anomaly entries from the benign ones. Intuitively, by representing the DP sensitivity as ``water level'' to hide people's heights, Figure~\ref{fig:3} depicts the two major phases of our approach, i.e., the \emph{learning} phase (top), and the \emph{prediction} phase (bottom). First, in the \emph{learning} phase, the data owner starts by sending to the analyst IID data excerpts protected under regular DP guarantees. Note an analyst who knows the shape of the pool and the DP guarantee of the top portion of the figure ``can'' infer a good guess of the distribution of the heights yet ``cannot'' accurately conduct anomaly detection.

Analyst then using all the DP data excerpts, computes a sufficiently accurate estimation of the distribution of the underlying dataset (to be used for effectively disentangling anomalous records from the benign ones) and sends it to the data owner. Second, in the \emph{prediction} phase (bottom of Figure~\ref{fig:3}), the data owner, using the approximated learnt PDF, disentangles the most likely benign records from the suspicious ones, and focuses the DP mechanism attention on them. To this end the accuracy of anomaly detection will be significantly boosted through ``entangling'' the privacy of data records with the learnt PDF. 


\begin{figure}[ht] \centering
	\includegraphics[width=1\linewidth, viewport=0 325 640 805,clip]{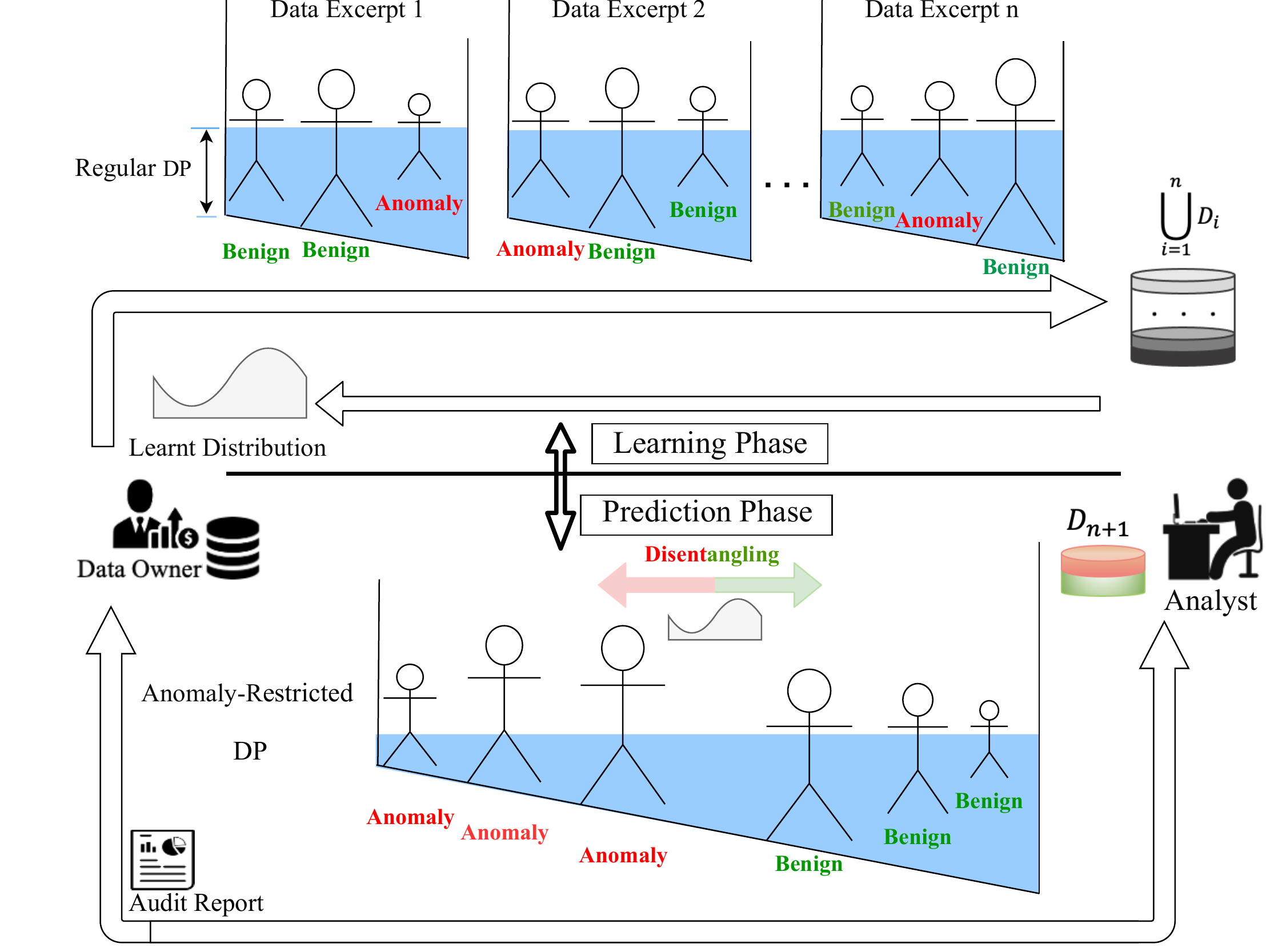}
\caption{The main idea of DPOAD which comprises learning and prediction phases. The water height denotes the sensitivity of DP mechanism to hide people's height. The learnt PDF is used to disentangle (likely) benign records from (likely) anomalous ones.}\vspace{-0.15in}
\label{fig:3}
\end{figure}


\subhead{Contributions} 
We design our {\it D}ifferentially {\it P}rivate {\it O}utsourcing of {\it A}nomaly {\it D}etection (\emph{DPOAD}) framework for this iterative approach in Section~\ref{sec:system} and theoretically analyze its privacy guarantee in Section~\ref{sec5}. Our analysis shows that this mechanism satisfies a popular version of DP called the Random Differential Privacy~\cite{Hall_Wasserman_Rinaldo_2013} which is widely used in protecting   datasets with unbounded/unknown sensitivities (similar to anomaly included datasets). In particular, we show that with DPOAD preserves $\epsilon$-DP of benign (anomalous) records with a probability close to (equivalently far from) one.
Furthermore, we formally benchmark DPOAD with the Laplace and Pain-Free mechanisms on different anomaly detection applications such as network, parking monitoring, IoT and credit card anomaly (fraudulent) detection. Finally, we experimentally evaluate the accuracy improvements in anomaly detection under DPOAD in comparison to the baselines Laplace mechanism and the Pain-Free mechanism~\cite{pmlr-v70-rubinstein17a}. In summary, our main contributions are as follows:
\begin{itemize}
\item As per our knowledge, DPOAD is the first approach that addresses the \emph{outsourced} version of the long-standing conflict between differential privacy and anomaly detection. This is especially relevant today as outsourcing computing tasks to clouds has seen a global surge lately. To that end, DPOAD introduces a novel iterative approach that allows the data owner and analyst to collectively find an optimal differential privacy setting under which the accuracy of anomaly detection can be maximized while guaranteeing privacy.

\item DPOAD tackles various technical challenges pertaining to privacy-preserving anomaly detection in the outsourced setting, including (a) learning data distribution from orthodox DP excerpts of the dataset (which are not directly suitable for anomaly detection), (b) guaranteeing an iterative version of DP, called RDP~\cite{hall2013random}, by sampling the sensitivity, first uniformly from the domain of the data (the learning phase), and then from the learnt distribution (the prediction phase), (c) further boosting the accuracy through refining the estimated distribution of the data over time using the computed anomaly scores, and leveraging this distribution into the sensitivity sampler.

\item Our experimental results demonstrate that DPOAD significantly improves the accuracy of  anomaly detection compared to the Laplace mechanism and the Pain-Free solution~\cite{pmlr-v70-rubinstein17a}). The improvement rate varies between [0-100]\% for the parking dataset, [100-500]\% for the credit card clients dataset, and [300-600]\% for network and building monitoring dataset.
\end{itemize}

The rest of the paper is organized as follows. Section~\ref{sec:model} provides necessary background, and Section \ref{sec:system} defines the DPOAD framework. Section~\ref{sec5} analyzes privacy, utility and efficiency advantage of DPOAD. Section~\ref{sec:exp} details the experiments. Section~\ref{relat} reviews the related work, and Section~\ref{sec:conclusion} concludes the paper.

\section{Background} 
\label{sec:model}

We review some notion and notations, as well as some required background on differential privacy for the theoretical foundations of the DPOAD framework. 

\subsection{Notions and Notations}

\subhead{Data~\cite{DBLP:conf/ccs/AsifPV19}} We consider a database as a multiset of elements from
a set $\mathcal X$, which is the set of possible values of records. In a database,
we assume each record is associated with a distinct individual. We represent a database $\mathcal{D}$ as a histogram in $\mathcal{H}=\{y \in \mathbb{N}^{\mathcal{X}} : ||y||_1 < \infty \}$,
where $\mathcal{D}$ is the set of all possible database, $\mathbb{N} = \{0, 1, 2, \cdots\}$, and $C_j$ is the number of records in $\mathcal{D}$ that are identical to $j$. In particular, at each iteration $i$, the data owner releases an excerpt of her/his dataset $D_i$ collected between time $t_{i-1}$ and $t_i$; converts it to histogram $H_i$ consist of a fixed set of bins $B_1, B_2, \cdots, B_m$ (also a subset of $\mathcal{H}$), and their corresponding counts reported in a time-window $C_{j\leq m ,t_{i-1}\leq t\leq t_i}$. The bins' edges (boundaries) are defined by the MSSP in a way that maximizes the detection accuracy (independent from the privacy mechanism), and the counts are collected in a specified time-window and are aggregated in a matrix whose elements are time series counts $C_{j\leq m ,t_{i-1}\leq t\leq t_i}$. This count matrix is also denoted by $C_i$ to save notations.

\subhead{Distance Metrics~\cite{NIPS2015_2b3bf3ee}}
For two probability distributions $p$ and $q$ over $[N]=\{1,\cdots, N\}$, for some $N \in \mathbb{Z}_{+}$. (1) the \textit{Kolmogorov} distance between $p$ and $q$ is defined as $d_K(p,q) \overset{\operatorname{def}}{=}	\max_{j\in[N]} |\sum_{i=1}^jp(i)-\sum_{i=1}^jq(i)|$, and (2) the total variation distance between $p$ and $q$ is defined as $d_{TV} (q, p)=0.5\cdot \norm{q-p}_1$ which is sometimes called \textit{statistical distance}. Note that $d_K(p,q)\leq d_{TV} (q, p)$.

\subhead{Agnostic Learning~\cite{NIPS2015_2b3bf3ee}} Let $\mathcal{C}$ be a family of distributions over a domain $\Omega$. Given sample access to an unknown distribution $p$ over $\Omega$ and $0<\alpha, \beta<1$, the goal of an ($\alpha, \beta$)-agnostic learning algorithm for $\mathcal{C}$ is to compute a hypothesis distribution $h$ such that with probability at least $1-\beta$, it holds $d_{TV} (h,p) \leq C \cdot opt_{\mathcal{C}}(p) + \alpha$; where $opt_{\mathcal{C}}(p) := \inf _{q \in \mathcal{C}} d_{TV} (q, p)$ and $C \geq 1$ is a
universal constant. 
\subsection{Differential Privacy}
\label{section: DP def}
\label{def: differential privacy original}

We follow the standard definition of
$\epsilon$-\DP~\cite{DworkNRRV09,NissimRS07}. Let $\D$ be a dataset of
interest, and $d$, $d'$ be two adjacent subsets of $\D$ where $d'$ can be obtained from $d$ by simply adding or subtracting the data of one individual. A randomization mechanism $\mathcal M: \D \times \Omega\to \R$ is $\epsilon$-differentially private if
for all $S \subset \R$, 
\begin{align}\label{eq: standard def approximate DP original}
\Prob(\mathcal M(d) \in S) \leq e^{\epsilon} \Prob(\mathcal M(d') \in S) \;\; 
\end{align}

If the inequality fails, then a leakage ($\epsilon$ breach) takes place. 
We recall below a basic mechanism that can be used to answer queries in an $\epsilon$-differentially private way.

While strong $\epsilon$-DP is ideal, utility may demand compromise. Precedent exists for weaker privacy, with the
definition of ($\epsilon,\delta$)-DP where with small probability $\delta$, $\epsilon$-DP is allowed to fail. Another relaxed notion is ($\epsilon,\gamma$)-Random Differential Privacy (RDP)~\cite{Hall_Wasserman_Rinaldo_2013}, where on all but a small $\gamma$-proportion of unlikely database pairs, strong $\epsilon$-DP holds.

\begin{defn} \label{def:parameters}
Randomized mechanism $\mathcal M_q: \D \times \Omega \to \mathbb R$ preserves ($\epsilon,\gamma$)-random differential privacy, at privacy level $\epsilon>0$ and confidence $\gamma \in (0, 1)$, if  $\Prob [\forall S \subset \R, \Prob (\mathcal M(d) \in S) \leq e^{\epsilon} \cdot \Prob (\mathcal M(d) \in S)] \geq 1-\gamma$, with the inner probabilities over the mechanism’s
randomization, and the outer probability over neighbouring $d, d' \in \D$ drawn from some $P^{n+1}$. 
\end{defn}

\label{section: basic mech}
\subsection{Laplace Mechanism} The Laplace mechanism modifies a numerical query result by adding zero-mean noise (denoted as $Lap(b)$) subjecting to a Laplace distribution with mean zero and scale parameter $b$. It has density
    $p(x;b)=\frac{1}{2b}exp(-\frac{|x|}{b})$ and variance $2b^2$. The following sensitivity concept plays an important role in the design of differentially private mechanisms~\cite{10.1007/11681878_14}.

\begin{defn}[Global Sensitivity]\label{defn: sensitivity}
The sensitivity of a query $q: \D \to \R$ is defined as
$\Delta q= \max_{d,d':\Adj(d,d')} |q(d) - q(d')|$~\cite{DworkNRRV09,NissimRS07}.
\end{defn}
    
\begin{thm}	\label{thm: Lap mech}
Let $q: \D \to \mathbb R$ be a query , $\epsilon>0$.
Then the  mechanism $\mathcal M_q: \D \times \Omega \to \mathbb R$ 
defined by $\mathcal M_q(d) = q(d) + w$, with $w \sim Lap(b)$, where $b \geq \frac{\Delta q}{\epsilon}$, is $\epsilon$-differentially private~\cite{10.1007/11681878_14}.
\end{thm}

\subsection{Pain-Free Algorithm~\cite{pmlr-v70-rubinstein17a}}
Rubinstein et al. (the Pain-Free solution) developed an RDP mechanism using a \emph{sensitivity sampler}, which approximates the global sensitivity with high probability. 

\begin{defn} [Approximate Sensitivity]
\label{defn:Appsens}
Consider sampling parameters size $1 \ll m \in \mathbb N$ and order statistic index $ k \leq m \in \mathbb N$, the (k,m)-approximate sensitivity of a query $q: \D \to \R$ is given by uniformly sampling $m$ sensitivity candidates, sorts them, and picks the $k^{th}$ one to calibrate the Laplace mechanism.
 \end{defn}
As shown in the following theorem, combined with generic mechanisms like Laplace, such a sampler enables systematizing of privatization.

\begin{thm} [Proof in \cite{pmlr-v70-rubinstein17a})]
\label{thm:painfree}
Consider any database $\D$ of $n$ records, privacy parameters $\epsilon>0$, $\gamma \in(0,1)$, Kolmogorov approximation confidence $0 <\rho <\min\{\gamma, 1/2\}$, and the known distribution $P$ on $\D$. The \textit{Pain-Free algorithm} defined as Laplace mechanism calibrated with (k,m)-approximate sensitivity preserves ($\epsilon,\gamma$)-random differential privacy, where $m=\left \lceil \cfrac{\log(1/\rho)}{2(\gamma-\rho)^2} \right \rceil$, $k=\left \lceil m\big(1-\gamma+\rho+\sqrt{\cfrac{\log(1/\rho)}{2m}} \right \rceil$, and $\rho= exp(W_{-1}(-\frac{\gamma}{2\sqrt{e}})+0.5)$.
\end{thm}
The derived specific expressions involve branches of the Lambert-W function, which is the inverse relation of the function $f(z) =z \cdot exp(z)$. Moreover, a number of natural choices for sampling distribution $P$ could be made. In particular, the Pain-Free algorithm and its privacy guarantee are derived by assuming a \textit{uniform} distribution defined over the domain of the dataset $\D$. However, the accuracy of the solution can be significantly boosted where a simulation process capable of approximating the actual $P$ exists, e.g., the anomaly scores computed by an MSSP can be leveraged.

\section{The DPOAD Framework}
\label{sec:system}

In this section, we introduce the main building blocks for our DPOAD
framework, provide a high level overview of our DPOAD framework and describe its main steps.

\subsection{The Building Blocks}
\label{sec:bb}
DPOAD relies on the following three modules; namely, the \emph{DP Distribution Learning}, the \emph{Monotonic Disentanglement} and the \emph{Count Reconstruction} which are described below. 
\subsubsection{DP Distribution Learning}
\label{sec:dplearn}
One important building block of DPOAD is the distribution learning module. This algorithm measures the closeness between distributions in total variation distance ($d_{TV}$), and the goal is to construct a hypothesis $h$ from the actual distribution $p$ such that (with high probability) the total variation distance $d_{TV}(h,p)$ between $h$ and $p$ is at most $\alpha$, where $\alpha > 0$ is called the approximation confidence. Since the accuracy of the sensitivity learning is directly linked to the soundness of the sensitivity estimation (which is crucial for providing a sufficient degree of protection), we consider two-phase operations for DPOAD, namely (1) \textit{the learning phase ($t<T$)} over which a sufficiently large number of samples are provided to MSSP (using the Pain-Free algorithm) to approximate the underlying distribution for a desired confidence, and (2) \textit{the prediction phase ($t\geq T$)} during which the distribution learning can be considered as sufficiently ``accurate'', and therefore, the learnt distribution with its approximation confidence can be successfully leveraged to the Pain-Free algorithm. The accuracy guarantee of the learnt discrete distribution from a sufficient number of DP data samples (sample complexity of $\mathcal{O}(|\mathcal{D}| \log(1/\beta)/(\epsilon \alpha))$) is quite high as proved in Theorem~\ref{dplearning}, and therefore, as shown in Section~\ref{sec:exp}, the value of $T$ is reasonably low (few iterations of communication). Figure~\ref{fig:pdfest} illustrates this theoretical results for data samples collected from a smart home environment in \cite{10.1145/3319535.3354254}. These phases are further described in the framework.

\begin{figure}[ht] \centering
	\includegraphics[width=0.6\linewidth]{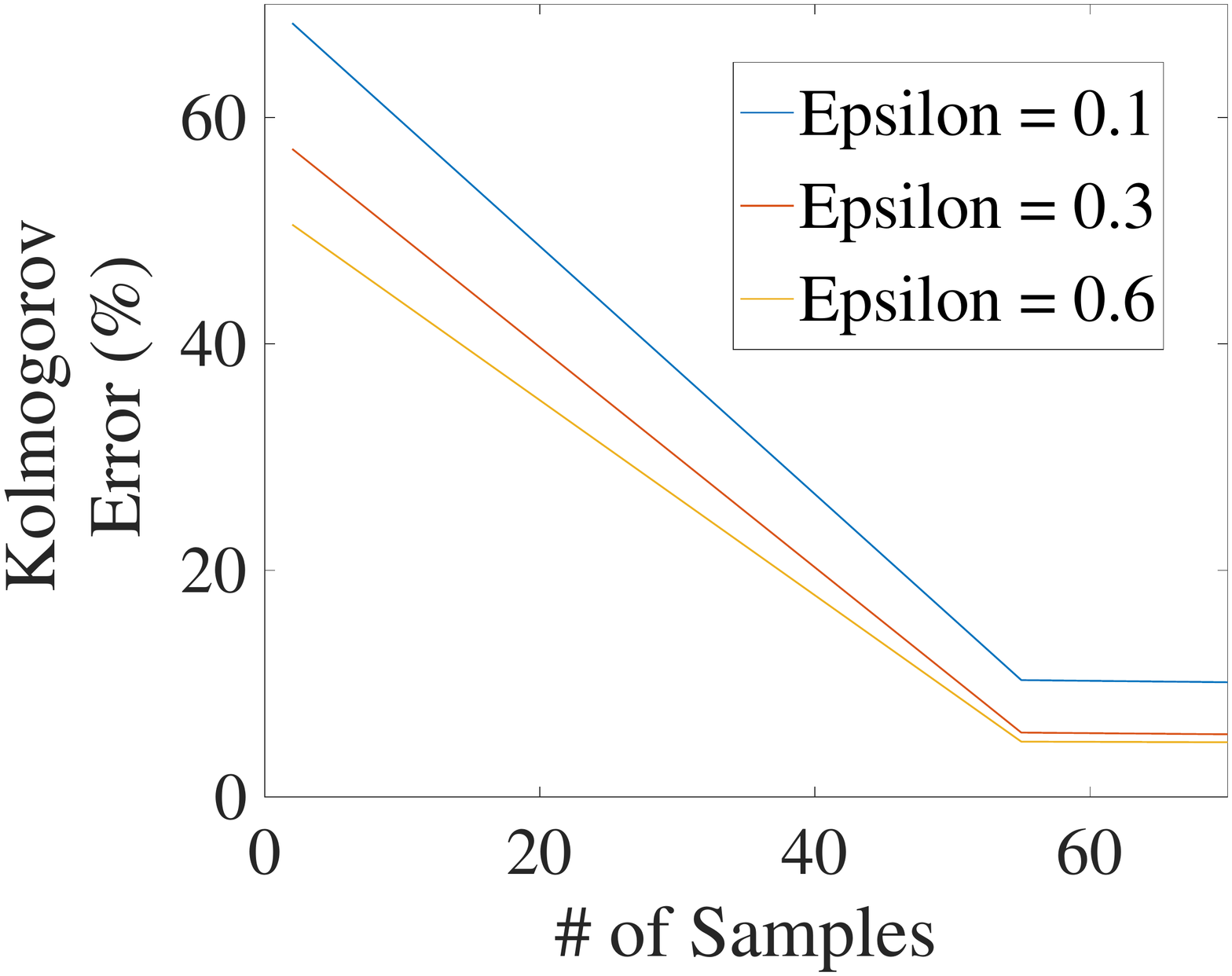}\vspace{-0.15in}
\caption{The average Kolmogorov error between the learnt distribution from DP samples and the actual data distribution, for an excerpt of the smart home dataset~\cite{10.1145/3319535.3354254}.}\vspace{-0.1in}
\label{fig:pdfest}
\end{figure}


\begin{figure}[ht] \centering
	\includegraphics[width=1\linewidth, viewport=17 480 545 669,clip]{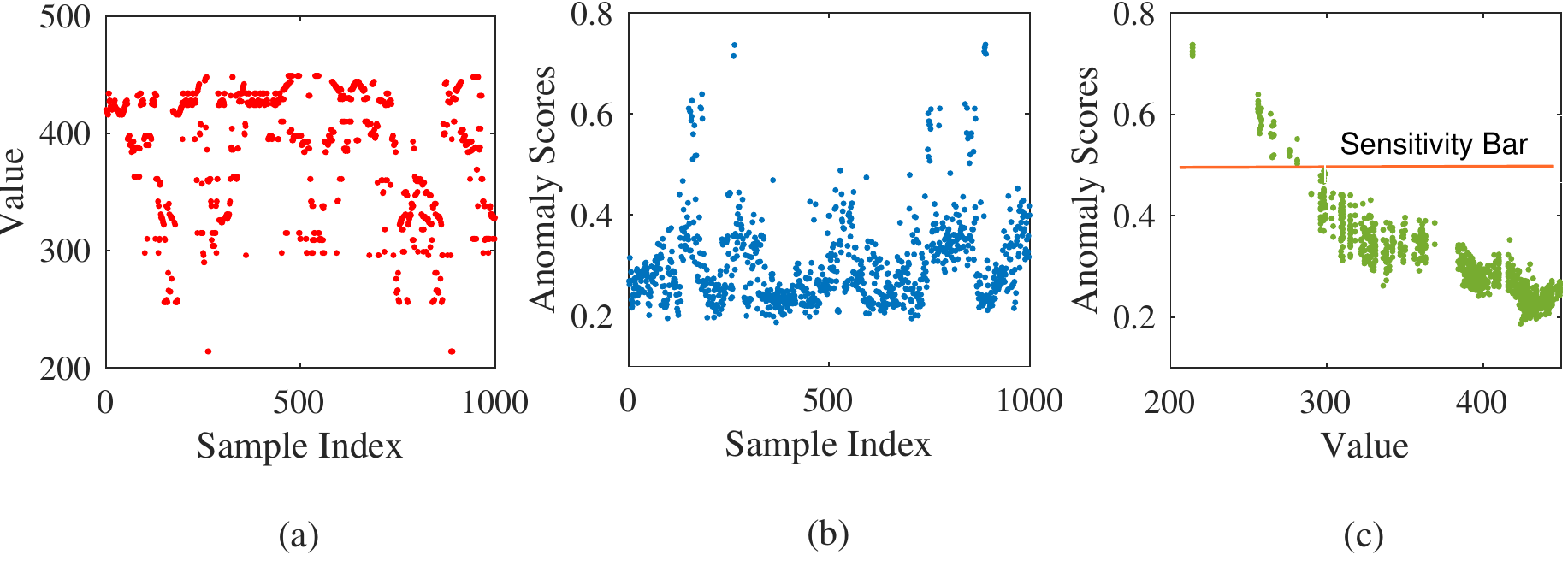}\vspace{-0.1in}
\caption{The density plot of the distribution of (a) the example data, and (b) its anomaly scores. Figure (c) shows that mapping data samples to their anomaly scores enables \textit{monotonic disentanglement}}
\label{densitmap}
\end{figure}

\subsubsection{Monotonic Disentanglement} 
\label{sec:disentang}
This building block is required during the prediction phase to enable effective anomaly detection with RDP. As shown in Figure~\ref{fig:3}, both small and large sensitivity values (the height of the water) entail weak anomaly detection results, specifically, the former introduces a large \textit{false positive} and the latter introduces a large \textit{false negative}. This issue is inherent to the notion of differential privacy with a constant or uniformly sampled sensitivity value which essentially makes the mechanism ineffective for many classification-based analyses including anomaly detection (see Section~\ref{sec:R$^2$DP as a
  Patch to Existing Work}). Therefore, an effective mechanism must provide some sort of \textit{quantified disentanglement} between benign and anomalous records (desired class of detection from the rest of the dataset) when imposing any notion of differential privacy. Fortunately, DPOAD can easily address this by leveraging the granular probability distribution of the data using the anomaly scores, meaning that the sensitivity to be sampled from a rather objective-oriented (anomaly detectable) probability distribution.

Beside disentangling, another issue lies in the fact that anomalous records can potentially take any value as shown in Figure~\ref{fig:3} and Figure~\ref{densitmap}. Therefore, sampling the sensitivity value, even merely from the benign records, may contribute to very large values. Therefore, effective disentanglement ideally maps the anomalous records to larger values so that tuning DP mechanism w.r.t. benign values could preserve both privacy of the benigns and accuracy of the anomaly detection. Figure~\ref{fig:3} metaphorically shows such a desirable \textit{monotonic} (larger anomaly score$\rightarrow$ less protection) disentanglement property that our DPOAD framework must prevail. Finally, Figure~\ref{densitmap} shows that a perfect \textit{monotonic} and \textit{disentangled} version of the original data could be generated through mapping the data records to their assigned anomaly scores. This finding is very intuitive because anomaly scores are essentially designed to monotonically disentangle the outlier records from the benign ones in a given dataset. We note that the fixed sensitivity bar shown in Figure~\ref{densitmap}(c) will look like the shape of water if we ``map back'' the \textit{monotonic} and \textit{disentangled} version of the data to its original format.

\subsubsection{Count Reconstruction} 
Clearly, the aforementioned mapping is crucial for DP to maintain the functionality of anomaly detection, and has to be applied all the time (both phases). However, since anomaly detection operates over counts, a count reconstruction function must be applied prior to sending the data over. The count reconstruction function is the reverse function of the monotonic disentangler which generates pseudo-counts $\Tilde{C_{i}}$ for the MSSP to compute the new set of anomaly scores.
\begin{figure*}[ht] \centering
	\includegraphics[width=0.95\linewidth, viewport=20 32 930 510,clip]{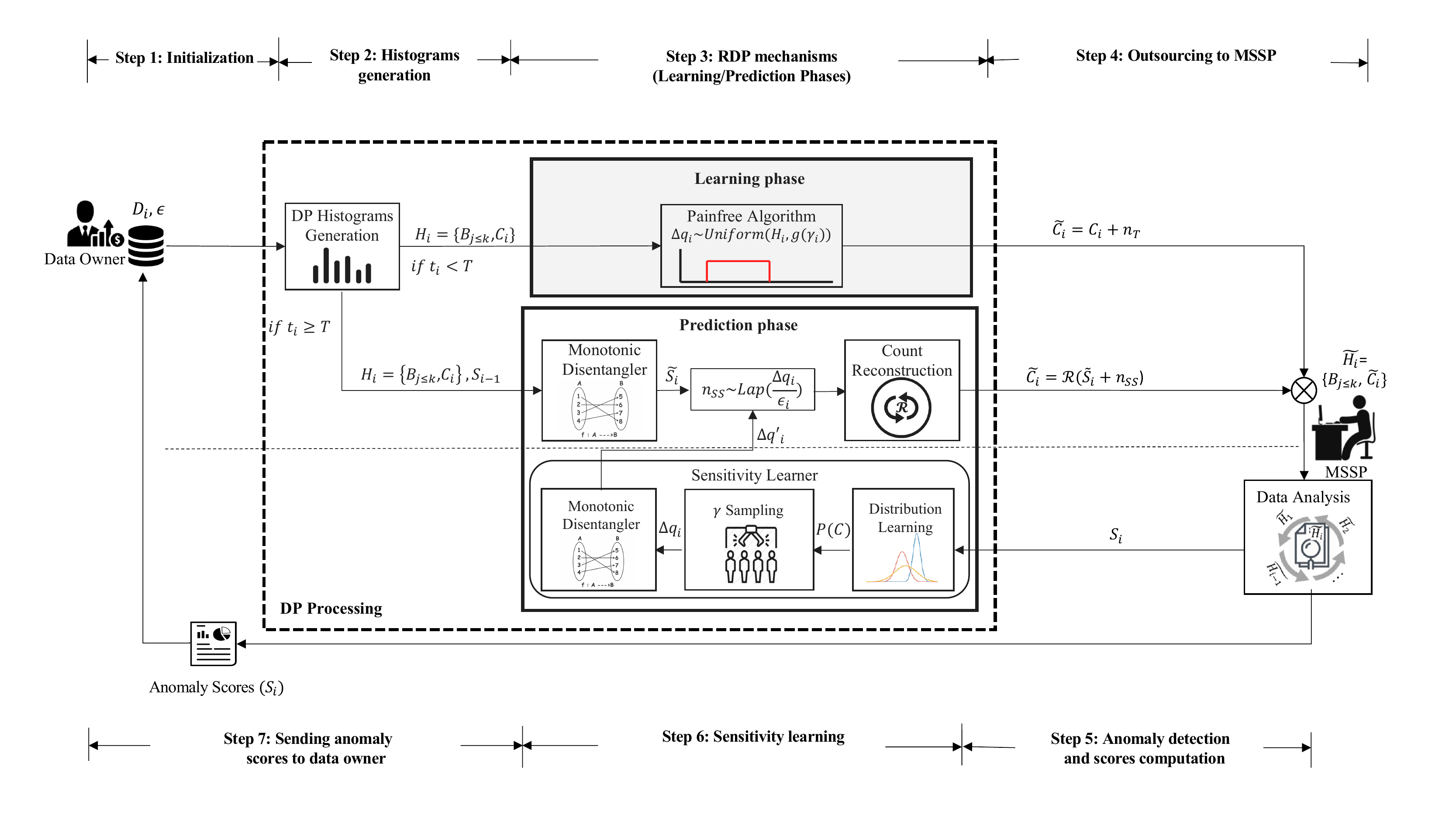}
\caption{The high level overview of the DPOAD framework}\vspace{-0.1in}
\label{fig:ov}
\end{figure*}
\subsection{The Framework}
As shown in Figure~\ref{fig:ov}, DPOAD includes the following steps.
\subsubsection{DP Outsourcing at the Data Owner Side}

\begin{description}
	\item[Step 1:] At iteration $i$, the data owner releases an excerpt of her/his dataset collected between time $t_{i-1}\leq t\leq t_i$; denoted by $D_{i\leq t}$; for anomaly detection analysis, it specifies the desired degree of differential privacy (DP) guarantee ($\epsilon$).
	
	\item[Step 2:] The data excerpt $D_i$ is converted into a histogram $H_i$ with a fixed set of bins $B_1, B_2, \cdots, B_k$. The bins' edges (boundaries) are defined by the MSSP in a way that maximizes the detection accuracy (independent from the privacy mechanism). 
	
	\item[Step 3:] Next, the corresponding count matrix $C_i=\{C_{\{j\leq k,  t_{i-1}\leq t \leq t_i\}} \\ \in \mathcal{C}\}$ consisting of time series counts of each bin will undergo the RDP mechanism. \footnote{Note that the privacy parameter $\gamma$ is the same as the confidence probability of anomaly detection. Our analysis and experiments show that running DPOAD with $\gamma \approx 0.2$ both converges and preserves $\epsilon$-DP of benign records.} This mechanism in the learning phase, prior to learning iteration threshold $T$ (detailed in Section~\ref{sec:dplearn}), is the Pain-Free algorithm which exerts a Laplace mechanism with a sensitivity value uniformly sampled from the domain of $\mathcal{C}$, and in the prediction phase consists of the \textit{monotonic disentangler} (described in Section~\ref{sec:disentang}), Laplace mechanism and the \textit{count reconstruction} function. In particular, the disentangler transforms the counts to their anomaly scores using the latest set of anomaly scores. The sensitivity of Laplace mechanism is provably computed using the \textit{sensitivity learner} (Step 6). 

	\item[Step 4:] The privatized counts $\tilde{C}_i$ is outsourced to the MSSP to compute the new set of anomaly scores.
          \end{description}
          \subsubsection{Sensitivity Update at the MSSP Side}
          \begin{description}
          \item[Step 5:] The MSSP executes a score-based anomaly detection function, e.g., the Kolmogorov–Smirnov test (KS test), over all the received pseudo-counts, and generates the anomaly scores $S_{i}$ corresponding to data of iteration $i$. 
          
          \item[Step 6:] The set of computed anomaly scores is analyzed using the sensitivity learner module to update the sensitivity value for the next iteration. The operations in this module includes (1) constructing the probability distribution of each histogram bin's counts ($P(\mathcal{C}))$) according to all previously aggregated data and their anomaly scores (the higher score of an observation is, the lower its probability mass function will become), (2) sampling $m$ sensitivities from the constructed distribution, sorting them and reporting the $k^{th}$ one as $\Delta q_i$ where $m$ and $k$ are given in the Pain-Free algorithm~\ref{thm:painfree}, and (3) monotonic disetanglement mapping function to generate an anomaly score version of the sensitivity $\Delta q'_i$.

	\item[Step 7:] The anomaly detection report for iteration $i$ is sent back to the data owner.
\end{description}

\IncMargin{1em}

\begin{algorithm}[!h]

\caption{Prediction Phase Algorithm in DPOAD}
 \small
 \SetKwInOut{Input}{Input}
 \SetKwInOut{Output}{Output}
\Input{Dataset $D$, Privacy budgets $\epsilon$, Number of samples ($n$) transmitted so far using Pain-Free, Learnt distribution $P$, Domain $\mathcal{C}$, Histogram function $\mathcal{H}$, number of histogram bins $k$}

\Output{Set of identified Anomalies $S_i$}

\textbf{Function}(Private Cloud)\\
$G\leftarrow P(\mathcal{C})$\\
$m \leftarrow I \gg m$ in Theorem~\ref{DPOADDP}\\
$\rho^{*} \leftarrow \cfrac{p}{(m+q)^{r}}$ \\
$k \leftarrow \left \lceil m\big(1-\gamma+\rho^{*}+\sqrt{-\cfrac{\log(1-\sqrt{1-\rho^{*}}))}{2m}} \right \rceil$\\

$ Sens \leftarrow \{\Delta q_{j\leq m} \sim G \}$\\
$Sens \leftarrow sort(Sens)$\\
$\Delta q_i\leftarrow Sens(k)$\\

$\Tilde{H}_{i}=\{B_{\{j\leq k \}},C_{\{j\leq k \}}\}\leftarrow \mathcal{H} (D_i,k)$\\

$\Tilde{S}_i \leftarrow P(C_{\{j\leq k \}})$\\

$\Tilde{S}_i \leftarrow \Tilde{S}_i+Lap(\frac{\Delta q_i}{\epsilon_i}) $\\

$\mathcal{R} \leftarrow normalize(P^{-1}(\cdot))$\\
\textbf{Function}(MSSP)\\
$\Tilde{C}_i \leftarrow \mathcal{R}(\Tilde{S}_i)$\\ 

$S_i=\mathcal{A}(\Tilde{H}_{1,2,\cdots, i})$: Analyzing all previously received datasets to optimally identify anomalous records\\

Update $P$

\textbf{return} $S_i, P$
 	\label{alg:analyst-actions1}
 \end{algorithm} 

\DecMargin{1em}
Since DPOAD during learning phase merely focuses on distribution learning using the Pain-Free algorithm, we summarize the above steps of DPOAD over the prediction phase in Algorithm~\ref{alg:analyst-actions1}. Next, we analyze the privacy and utility of DPOAD.

\section{Analysis}
\label{sec5}
In this section, we analyze privacy, utility and efficiency advantages of DPOAD.
\subsection{Privacy Analysis}
\label{subsec:privacy}

Our analysis adapts the well-known findings in \text{private distribution estimation} \cite{NIPS2015_2b3bf3ee}, which are reviewed as below. Sepecifically, let $\mathcal{C}_N$ be the family of all distributions over $[N]$, the following holds:
\begin{thm}
\label{dplearning}
There is a computationally efficient $\epsilon$-differentially private $(\alpha, \beta)$-learning algorithm for $\mathcal{C}_N$ that uses $n = \mathcal{O}((N + \log(1/\beta))/\alpha^2 + N \log(1/\beta)/(\epsilon \alpha))$ samples.
\end{thm}

\subhead{Sketch Proof of Theorem \ref{dplearning}}
Diakonikolas et al.~\cite{NIPS2015_2b3bf3ee}, by properties of the Laplace distribution combining with the VC inequality~\cite{devroye2012combinatorial} and the Dvoretzky-Kiefer-Wolfowitz (DKW) inequality~\cite{10.2307/2237374}, showed that at least $\mathcal{O}(N + \log(1/\beta) / \alpha^2+ N \log(1/\beta)/(\epsilon \alpha^2))$ number of $\epsilon$-DP data samples is required to accurately (($\alpha,\beta$)-learning) estimate the non-private discrete probability distribution of the dataset.
 

\subsubsection{Privacy Guarantee of DPOAD}
We now present the required conditions for DPOAD to satisfy a desired level of ($\epsilon,\gamma$)-RDP. 

\begin{thm}
\label{DPOADDP}
   The DPOAD framework satisfies ($\epsilon,\gamma$)-RDP if for the sensitivity sampling parameters size $m$ , $k$ we have
    $$
m=\begin{cases}
\left \lceil \frac{\log(1/\rho)}{2(\gamma-\rho)^2} \right \rceil, & \text{Learning}\\
            \left \lceil \frac{-2\log(1-\sqrt{1-\rho})}{(T-\sqrt{T^2-4})^2} \right \rceil & \text{Prediction,}
		 \end{cases}
$$
    $$
k=\begin{cases}
\left \lceil m\Big(1-\gamma+\rho+\sqrt{\frac{\log(1/\rho)}{2m}} \Big ) \right \rceil, & \text{Learning}\\
           \left \lceil m\Big(1-\gamma+\rho+\sqrt{-\frac{\log(1-\sqrt{1-\rho}))}{2m}} \Big) \right \rceil, & \text{Prediction}
		 \end{cases}
$$
where $T=2+\frac{n\cdot \epsilon_i}{2N\cdot c}$, and $c$ is the distribution learning constant parameter for Laplace mechanism. Moreover, the $\rho^{*}$ values which minimize $k$ for a fixed $m$ in each state are
$$
\rho^{*}= \begin{cases}
 exp(W_{-1}(-\frac{\gamma}{2\sqrt{e}})+0.5), & \text{Learning}\\
            \frac{p}{(m+q)^{r}}, & \text{Prediction.}
		 \end{cases}
$$
\end{thm}
where $p\approx 1.426$, $q\approx 0.8389$ and $r \approx 0.4589$.
\begin{proof}
In the learning state, the ($\epsilon,\gamma$)-differentially private Pain-Free algorithm is applied and the results in the Theorem are borrowed from~\cite{pmlr-v70-rubinstein17a}, so it remains to prove the privacy guarantee of the prediction phase. In the prediction phase ($t \geq T$), let $\Prob_{\mathcal{C}}$ be the latent distribution for the count values of a given histogram bin $B_j\leq k$. In line with the sensitivity analysis given in~\cite{pmlr-v70-rubinstein17a}, for the i.i.d. samples of sensitivity $\Delta q_1, \cdots, \Delta q_m$ drawn from $\mathcal{C}$, denote the corresponding fixed unknown CDF, and corresponding random empirical CDF, by
\begin{eqnarray*}
\label{discretizePDF}
\Phi(x) = \Prob_{\mathcal{C}}(\Delta q\leq x) \\
 \Phi_m(x) = \frac{1}{m} \sum^m_{i=1}\mathds{1}(\Delta q_i\leq x). 
\end{eqnarray*}
However, DPOAD's samples are drawn from DP estimated distribution $\hat{\Prob}_{\mathcal{C}}$ learnt over a dataset $S$ of $n$ noisy samples received by time $t$. 
Let $\hat{\Delta} q_{\{i\leq m\}}$ be the ascending sorted sample sensitivities from $\hat{\Prob}_{\mathcal{C}}$, and $\hat{\Phi}_m(x)=\frac{1}{m}\sum^m_{i=1}\mathds{1}(\hat{\Delta} q_i\leq x)$  be its corresponding random empirical CDF. Moreover, let $\hat{\rho} \in (0, 1)$ be the DKW probability error parameter, provided that 
\begin{equation}
   1 -\gamma+ \rho+\hat{\rho} \leq 1 \Leftrightarrow \rho' \leq \gamma-\rho,
\end{equation}
then for the random sensitivity $\hat{\Delta}=\hat{\Delta} q_k$, where $k = \ceil{m(1-\gamma + \rho + \hat{\rho})}$, we have $\hat{\Phi}_m(\hat{\Delta}) \geq 1 -\gamma+ \rho+\hat{\rho}$. Suppose $\tilde{C}_t$ is the released set of DP counts generated by randomizing by the current set of $C_t \in \Prob_{\mathcal{C}}$ in accordance with the sample sensitivity $\hat{\Delta}$. Define the events

\begin{eqnarray*}
& A_{\hat{\Delta}} = \{\forall \R \subset \mathbb{R}^n, \hat{\Prob}(\tilde{C}_t\in \R) \leq e^{\epsilon} \cdot \hat{\Prob}(\tilde{C'}_t\in \R)\}\\
& B_{\rho'}= \{\sup \limits_{\hat{\Delta}}(\hat{\Phi}_m(\hat{\Delta}) - \Phi(\hat{\Delta})) \leq \rho'\}
\end{eqnarray*}
The first is the event that DP holds for a specific DB pair,
when the mechanism is run with (possibly random) sensitivity parameter $\hat{\Delta}$; the second records the empirical CDF
uniformly one-sided approximating the CDF to level $\rho'$. Moreover, by definition of differential privacy, we have
\begin{equation}
   \forall \Delta>0 , \ \ \ \ \Prob_{\{C_t,C'_t \in \Prob_{\mathcal{C}} \}} (A_{\Delta}) \geq \Phi (\Delta). 
\end{equation}

Putting inequalities (2) and (3) together yields that
\begin{align*}
   &\Prob_{\{C_t,C'_t,\hat{\Delta} q_1, \cdots, \hat{\Delta} q_m \}}(A_{\hat{\Delta}})= \\
  &\mathbb{E}(\mathds{1} [A_{\hat{\Delta}}]|B_{\rho'}) \Prob(B_{\rho'})+\mathbb{E}(\mathds{1} [A_{\hat{\Delta}}]|\Bar{B}_{\rho'}) \Prob(\Bar{B}_{\rho'})\\  
&\geq \mathbb{E}(\hat{\Phi}(\hat{\Delta})|B_{\rho'}) \Prob(B_{\rho'})    \\
&\geq \mathbb{E}(\hat{\Phi}_m(\hat{\Delta})-\rho'|B_{\rho'}) \Prob(B_{\rho'})
\end{align*}
\begin{lem}
For all $\rho'>0$, $\Prob(B_{\rho'})$ in the DPOAD methodology is tightly lower-bounded by $\left[1-exp\left(-m\left(Q-\sqrt{Q^2-4}\right)^2/2\right)  \right]^2$, where $Q=2+\frac{n\cdot \epsilon}{2N\cdot m \cdot c}$ and $c$ is the hidden constant term in $\mathcal{O}(N + \log(1/\beta) / \alpha^2+ N \log(1/\beta)/(\epsilon \alpha^2))$.
\end{lem}

\begin{proof}
Using triangle inequality, we have
\begin{align*}
   &\Prob(B_{\rho'})=\Prob\left(\{\sup \limits_{\hat{\Delta}}(\hat{\Phi}_m(\hat{\Delta}) - \Phi(\hat{\Delta})) \leq \rho'\}\right)\\
   &=\Prob\left(\{\sup \limits_{\hat{\Delta}}\left(\hat{\Phi}_m(\hat{\Delta}) - \hat{\Phi}(\hat{\Delta})+\hat{\Phi}(\hat{\Delta}) -\Phi(\hat{\Delta})\right) \leq \rho'\}\right)\\
    &\geq \Prob\left(\{\sup \limits_{\hat{\Delta}}(\hat{\Phi}_m(\hat{\Delta}) - \hat{\Phi}(\hat{\Delta}))\}\leq \rho'-\{\sup \limits_{\hat{\Delta}}(\hat{\Phi}(\hat{\Delta})-\Phi(\hat{\Delta})) \}\right)
\end{align*}
The necessary condition for the last expression, denoted by $E$, to be a non-zero probability is that we have $\rho'>\sup \limits_{\hat{\Delta}} \{\hat{\Phi}(\hat{\Delta})-\Phi(\hat{\Delta}) \}$. To tightly calculate $E$, our idea is to assume that we have $\sup \limits_{\hat{\Delta}} \{\hat{\Phi}(\hat{\Delta})-\Phi(\hat{\Delta}) \}\leq x.\rho'$ (event A) for a real-value $x<1$. We can then write $E\geq \Prob(A)\cdot P(B|A)$, where $B|A$ is the event that $\sup \limits_{\hat{\Delta}} \{\hat{\Phi}_m(\hat{\Delta}) - \hat{\Phi}(\hat{\Delta})\} \leq (1-x).\rho'$.
Moreover, from Theorem~\ref{dplearning}, we can easily show that for a set of $n$ $\epsilon$-DP data samples, we have
\begin{align}
\label{re:ineq1}
  &\Prob(A)\geq 1- exp\left(\frac{N \epsilon c-n \epsilon ((1-x).\rho')^2}{c \epsilon +Nc(1-x).\rho'}\right)  
\end{align}
Moreover, using DKW inequality we have that for all $\rho'\geq \sqrt{\log{2}/(2m)}$, 

\begin{align}
\label{rd:ineq2}
&\Prob(B|A)=\Prob\left(\sup \limits_{\hat{\Delta}} \{\hat{\Phi}_m(\hat{\Delta}) - \hat{\Phi}(\hat{\Delta})\}\leq (1-x) \cdot \rho'\right) \\
& \geq 1-exp\left(-2m(1-x\rho')^2\right) \nonumber
\end{align}
Using the inequality of arithmetic and geometric means, the multiplication of the two probability lower-bounds~\ref{re:ineq1} and \ref{rd:ineq2} together is maximized when the two terms \ref{re:ineq1} and \ref{rd:ineq2} are equal. Solving this equality is equivalent to solving the equality between the two exponent powers, which by some algebraic manipulation and ignoring negligible terms  returns $x=\frac{Q-\sqrt{Q^2
-4}}{2\rho'}$ where $Q=2+\frac{n\cdot \epsilon}{2N\cdot m\cdot c}$. Consequently, the maximum lower-bound of $\Prob(B_{\rho'})$ is
\begin{equation}
\left[1-exp\left(-m\left(Q-\sqrt{Q^2-4}\right)^2/2\right)  \right]^2
\end{equation}
\end{proof}
However, $Q$ has functionality of $m$. To get rid of this functionality, we can introduce the following weaker lower-bound for $\Prob(B_{\rho'})$
\begin{equation}
\left[1-exp\left(-m\left(T-\sqrt{T^2-4}\right)^2/2\right)  \right]^2
\end{equation}
where $T=2+\frac{n\cdot \epsilon}{2N\cdot c}$. Therefore, 
\begin{align*}
   &\Prob_{\{C_t,C'_t,\hat{\Delta} q_1, \cdots, \hat{\Delta} q_m \}}(A_{\hat{\Delta}})\\
&\geq \mathbb{E}(\hat{\Phi}_m(\hat{\Delta})-\rho'|B_{\rho'})\\& \geq (1-\gamma-\rho+\rho'-\rho')\left[1-exp\left(-m\left(T-\sqrt{T^2-4}\right)^2/2\right)  \right]^2\\& \geq (1-\gamma-\rho)(1-\rho)\geq (1-\gamma+\rho-\rho)= 1-\gamma
\end{align*}
The last inequality follows from $\rho<\gamma$; the penultimate inequality follows from $m\geq
            \left \lceil \frac{-2\log(1-\sqrt{1-\rho})}{(T-\sqrt{T^2-4})^2} \right \rceil $. Therefore,
 $$m=\begin{cases}
\left \lceil \frac{\log(1/\rho)}{2(\gamma-\rho)^2} \right \rceil, & \text{learning~\cite{pmlr-v70-rubinstein17a}}\\
            \left \lceil \frac{-2\log(1-\sqrt{1-\rho})}{(T-\sqrt{T^2-4})^2} \right \rceil & \text{Prediction}
		 \end{cases}
$$Finding the expression for $k$, we note that since both $0<x<1$ then   
         \begin{align*}
  & \left[1-e^{\{-2m\rho'^2\}}\right]^2 \geq \left[1-exp\left(-m(T-\sqrt{T^2-4})^2/2\right)  \right]^2\\
  &\geq 1-\rho
\end{align*}   
By some straightforward algebraic manipulation, we have $\rho'\geq \sqrt{-\frac{\log(1-\sqrt{1-\rho}))}{2m}}$, and hence, 
 $$           
k=\begin{cases}
\left \lceil m\big(1-\gamma+\rho+\sqrt{\frac{\log(1/\rho)}{2m}} \right \rceil, & \text{learning~\cite{pmlr-v70-rubinstein17a}}\\
           \left \lceil m\big(1-\gamma+\rho+\sqrt{-\frac{\log(1-\sqrt{1-\rho}))}{2m}} \right \rceil, & \text{Prediction}
		 \end{cases}
$$
\vspace{0.2in}

Finally, to find the optimal $\rho^{*}$ which minimizes $k$ for a fixed $m$ in the prediction phase, the following is the best computed fit function.
$$
\rho^{*}\approx \begin{cases}
 exp(W_{-1}(-\frac{\gamma}{2\sqrt{e}})+0.5), & \text{Learning~\cite{pmlr-v70-rubinstein17a}}\\
            \frac{1.426}{(m+0.8389)^{0.4589}}, & \text{Prediction.}
		 \end{cases}
$$
\end{proof}

\subsection{On the Relationship Between RDP and DP}
\label{RDP vs DP}
\begin{thm}
A Laplace mechanism which calibrates (k,m)-approximate sensitivity $\hat{\Delta}=\hat{\Delta} q_k$, where $k = \ceil{m(1-\gamma + \rho + \hat{\rho})}$, preserves $\epsilon$-DP for all benign records with probability 1.
\end{thm}
\begin{proof}
We have $\hat{\Phi}_m(\hat{\Delta}) \geq 1 -\gamma+ \rho+\hat{\rho}$. Therefore, for all records where $\gamma \leq \rho + \hat{\rho})$, we have $\hat{\Phi}_m(\hat{\Delta})=1$ and $k=m$. Such records are those who have been frequently observed (non-anomalous) and as a result their estimation error are smaller than $\gamma$. 
\end{proof}

\subsection{Utility Analysis}
\label{subsec utility}
Motivated by the analysis presented by Asif et al.~\cite{DBLP:conf/ccs/AsifPV19}, to formally analyze the utility pay-off of DPOAD, we employ the following abstraction of the process in privacy preserving anomaly detection algorithms. Suppose the anomaly identification algorithm is a binary functions $f : D \rightarrow \{0, 1\}$, and let $\mathcal{M}$ be the DP mechanism that is supposed to compute $f$. Now let $x$, $y$ be two databases that differ in one record and that $f (x) = 0$ and
$f (y) = 1$. Asif et al.~\cite{DBLP:conf/ccs/AsifPV19} have shown that the differential privacy constraints create a trade-off, i.e., one of the two following expressions always holds. 
\begin{align}
 & \Prob(\mathcal{M}(x) \neq f (x)) \geq 1/(1 + e^{\epsilon})
\ \ \text{or} \nonumber \\
 &\Prob(\mathcal{M}(y) \neq f (y)) \geq 1/(1 + e^{\epsilon}) \nonumber. 
\end{align}
which means that with an intensity depending on the value of $\epsilon$ (the higher the privacy
requirements are, i.e. for smaller $\epsilon$ the stricter this trade-off is), any DP anomaly identification solution always entails high rate of either false positive or false negative. As a result, this trade-off leads to a very low precision rate for DP anomaly detection algorithms~\cite{DBLP:conf/ccs/AsifPV19}. In the following theorem, we show that by virtue of
 clever sensitivity adjustment and monotonic disentanglement, DPOAD can vastly improve on this trade-off.
\begin{thm}
\label{dpoadutility}
Precision of the DPOAD algorithm is at least 
\begin{align}
  \frac{2+e^{-\epsilon ((m/k)+1)}-e^{-\epsilon}-2e^{-\epsilon m/k}}{2+e^{-\epsilon ((m/k)+1)}-e^{-\epsilon m/k}-2e^{-\epsilon}}  
\end{align}
times better than anomaly detection using Laplace mechanism.
\end{thm}

\begin{proof}
From~\cite{DBLP:conf/ccs/AsifPV19} (Claim 1 and Figure 1), we have $\Prob(\mathcal{M}(x) \neq f (x))\propto e^{\epsilon} (\Prob(\mathcal{M}(y) \neq f (y)))$. On ther hand, from Theorem~\ref{DPOADDP}, we can tune DPOAD such that $k$ minimizes w.r.t a fixed $m$. Since DPOAD applies monotonic disentanglement on the data values, the quantile level $k/m$ can be seen as an indication of the probability of $\Prob(\mathcal{M}(y) \neq f (y))$. 
As a result, we have
\begin{align*}
& \Prob(\mathcal{M}(x) \neq f (x))\propto e^{\epsilon} (\Prob(\mathcal{M}(y) \neq f (y))), \ \ \ \ \ \ \ \ \ \ \ \text{Laplace}\\
    &       \Prob(\mathcal{M}(x) \neq f (x))\propto e^{\epsilon m/k} (\Prob(\mathcal{M}(y) \neq f (y))), \ \ \ \ \ \  \text{DPOAD}
		 \end{align*}

Denoting by $0\leq p,q\leq 1$, respectively, the true negative and the true positive (complements of the above) probabilities, we have
\[Precision:
\frac{q}{1-p+q}\leq\begin{cases}
\frac{1-e^{-\epsilon}(1-p)}{1-p+1-(e^{-\epsilon}(1-p))} & \text{Laplace}\\
\frac{1-e^{-\epsilon m/k}(1-p)}{1-p+1-(e^{-\epsilon m/k}(1-p))} & \text{DPOAD}
		 \end{cases}
		 \]
It is easy to verify that for all choices of $p$, the DPOAD precision is placed higher because $m/k \geq 1$ and the expression is strictly increasing w.r.t. $m/k$. Hence, let $x=1-p$ and denote by $\mathsf{R}$ the precision ratio between the later and the former, we can easily show that this ratio w.r.t. $x$ is strictly increasing.  
\begin{align*}
  \mathsf{R} \geq \frac{2+e^{-\epsilon ((m/k)+1)}-e^{-\epsilon}-2e^{-\epsilon m/k}}{2+e^{-\epsilon ((m/k)+1)}-e^{-\epsilon m/k}-2e^{-\epsilon}}
\end{align*}
\end{proof}
\vspace{-0.03in}
where the expression on the left is ``the worst case'' improvement achieved at $x=1$. This bound is extremely loose lower-bound since we have assumed a ``linear relation'' between $k/m$ and strictness of the FP/FN tradeoff, while DPOAD uses a more granular PDF for more accurate TP predictions as shown in our experiments. 
\subsection{Discussion}
\label{sec:R$^2$DP as a
  Patch to Existing Work}
  
Finally, we discuss some important aspects of DPOAD which are derived from our results in Sections~\ref{sec5} and \ref{sec:exp}.


\subhead{Tuning $n$ in the Learning Phase}
What number of samples is sufficient to stop the learning phase and switch to the prediction phase? Parameter $n$, as shown in Theorem~\ref{DPOADDP}, affects parameters $k$ and $m$, which ultimately impact the quality of anomaly detection. In particular, if in the learning phase the underlying distribution is not learnt accurately due to a small number of samples ($n \downarrow \rightarrow T\downarrow \rightarrow m \downarrow$), we have $k \approx m$ which results in inaccurate anomaly detection (refer to Utility Analysis~\ref{subsec utility}). Therefore, it is important to ensure that a sufficient number $n$ of samples is aggregated by the MSSP to certifiably benefit from accurate anomaly detection. Figure~\ref{kmplot} (b) depicts a picture of the utility pay-off of DPOAD for various privacy parameters if a sufficiently large number of samples is used for learning the distribution of the dataset. 

\subhead{Tuning the Sensitivity Quantile Levels $k/m$}
 Figure~\ref{kmplot} (a) depicts the sensitivity quantile levels $k/m$ in both the learning and the prediction phases for different RDP parameter $\gamma$. This figure clearly shows this important parameter in the learning phase (Pain-Free) is negligibly better than in the prediction phase. This is expected since in the prediction phase, sensitivity sampling is conducted according to a learnt distribution over noisy observations which obviously entails less confidence and consequently, prevails a more conservative quantile value. Nonetheless, our experiments show that it is good enough for DPOAD to operate significantly better than the Pain-Free thanks to its monotonic disentangling function.
 
 \vspace{-0.15in}
 
 \begin{figure}[!tbh]
	\centering
	\subfigure[]{
		\includegraphics[width=0.45\linewidth, viewport=17 215 550 602,clip]{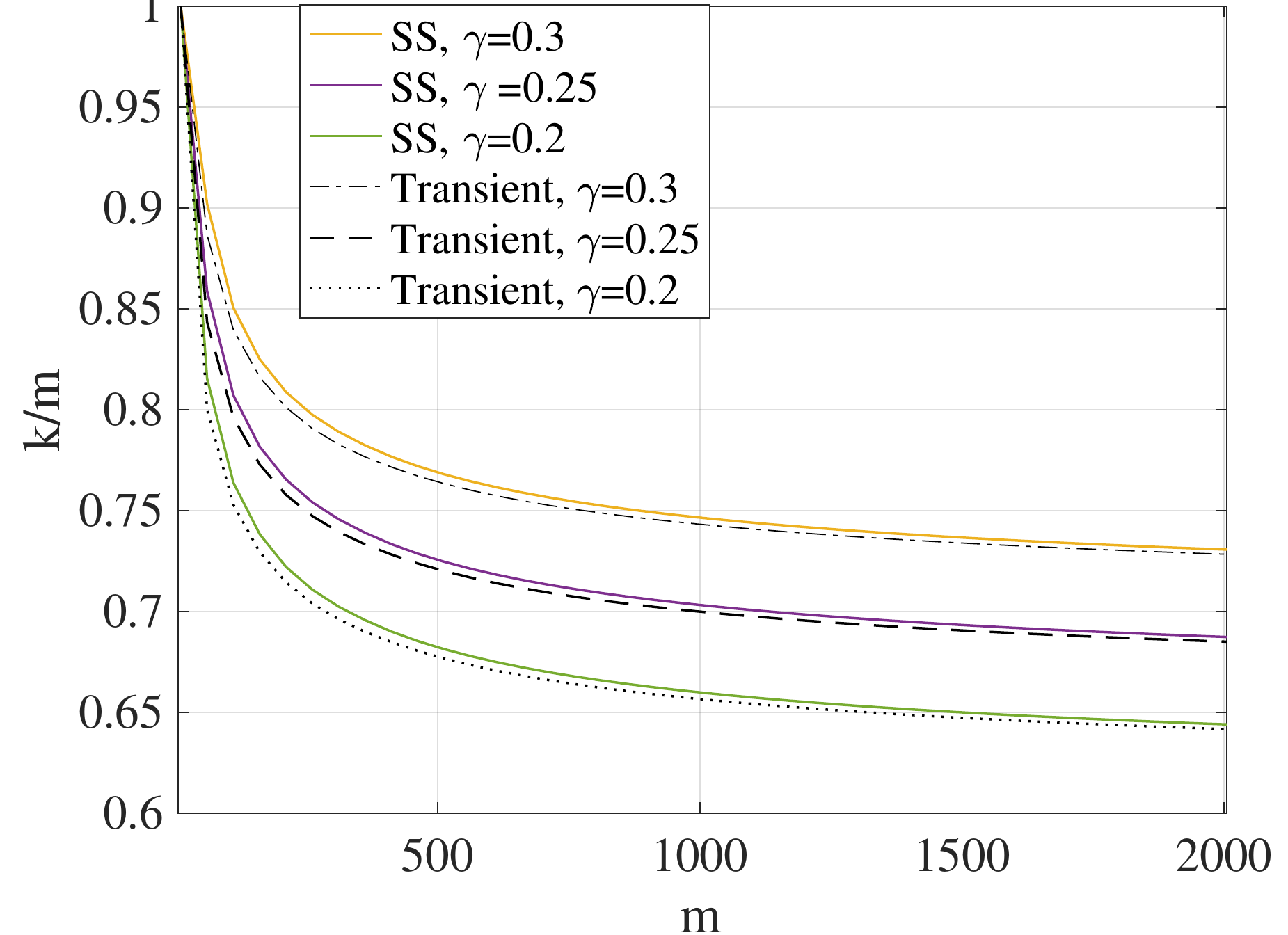}
		\label{fig:5Pre} }
	   \hspace{-0.15in}
	\subfigure[]{
		\includegraphics[width=0.45\linewidth, viewport=17 180 545 601,clip]{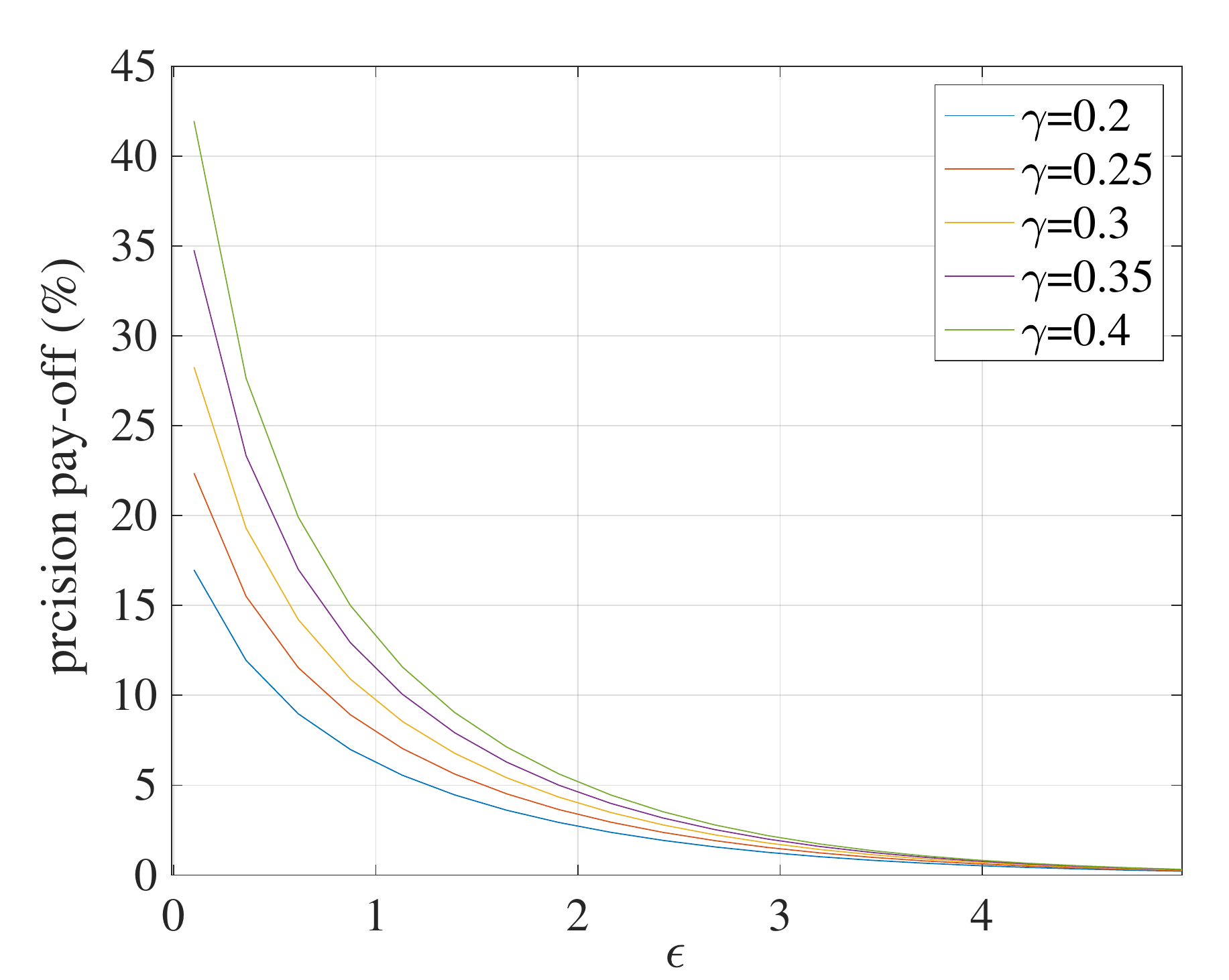}
		\label{fig:5Re}}
		\vspace{-0.15in}
	\caption{The (a) sensitivity quantile level $k/m$ in the Prediction and the Learning (Pain-Free) phases, and (b) precision pay-off (\%) in DPOAD for different $\epsilon$ and $\gamma$ parameters.}\vspace{-0.15in}
	\label{kmplot}
\end{figure}


\vspace{0.1 in}
\subhead{DPOAD in Other Statistical Inferences}
\label{Application}

Besides application of DPOAD in many use cases including IoT, fraud detection and resource/budget allocations (see the experiments), it can be used to compute various interesting data analysis functions in the outsourcing setting. In particular, DPOAD can be deployed to enable a specific type of analysis through iteratively constructing a monotonic disentangled version of the dataset according to what algorithm demands. For instance, in the heavy hitter (HH) computation, DPOAD can be applied with minor modifications. Specifically, MSSP at each iteration computes the distribution of the dataset and sends the heavy hitter scores back to the data owner. The monotonic disentangler can be tuned using these scores to preserve the privacy of all non-HH entries while enabling accurate HH identification. Other potential algorithms include SQL queries, clustering and optimal distribution learning.

\subhead{DPOAD and Non-Interactive DP}
Instead of histogram, we can apply DPOAD to the DP raw data publishing algorithms, namely non-interactive solutions. Unfortunately, the state-of-the-art is still struggling to find accurate and scalable non-interactive solutions and most of existing approaches usually result in significant utility loss~\cite{chen2011publishing}. 
DPOAD can be applied to boost utility for practical non-interactive solutions.

\subhead{Efficiency Analysis}
Table~\ref{relatedss1} summarizes the computational and communication complexity for the main action items of DPOAD. 
\vspace{-0.3in}
\begin{table}[ht]
	\caption{Complexities of DPOAD's main functions}\vspace{-0.1in}
	\label{relatedss1}
	\centering
	\begin{adjustbox}{max width=3in}
    \LARGE
		\begin{tabular}{|c|c|c|}
			\hline
			Function & Computation & Communication \\
			\hline

			\multirow{2}{*}{Distribution Learning} & \multirow{2}{*}{$O(n\log{N})$} & $\mathcal{O}((N + \log(1/\beta))/\alpha^2 $\\
	&&	$+ N \log(1/\beta)/(\epsilon \alpha))$	\\\hline
 Monotonic Disentanglement & $O(n log^{n})$& ---\\
			\hline
Pain-Free Algorithm &$O(n log^{n})$& ---\\
\hline
Anomaly Detection (KS)& $O(n^2)$ & ---\\
\hline
		\end{tabular}
	\end{adjustbox}
\end{table}


\section{Experimental Evaluations}
\label{sec:exp}


In this section, we experimentally evaluate the performance of DPOAD using a diverse collection of well-known datasets, including the parking Birmingham dataset \cite{DBLP:conf/smartct/StolfiA017}, individual household electric power consumption dataset \cite{electricCon}, credit card clients dataset \cite{DBLP:journals/eswa/YehL09a} and KDDCup99 dataset \cite{KDD99cup}. 

More specifically, the parking dataset consists of the occupancy rates of 31 parking lots in Birmingham for 2 months in 2016. The individual household electric power consumption dataset contains around 2 million power consumption measurements in 47 months from 2006 to 2010 in France.
The credit card clients dataset consists of customers default payments in Taiwan in 2005. The KDDCup99 dataset consists of ``good'' and ``bad'' network connections simulated in a military network environment. The anomaly detection in our experiments is based on the histogram combination of all the attributes in every dataset. Table~\ref{Table:dataset} summarizes the characteristics of the datasets.

\vspace{-0.1in}

\begin{table}[!h]
\small
\caption{Characteristics of the experimental datasets}\vspace{-0.1in}
\centering
\begin{adjustbox}{width=0.45\textwidth,center}
\begin{tabular}{|l|c|c|}
\hline
\bfseries{Databset}	&	\bfseries{Size}	&	\bfseries{\# of Attributes}		\\ \hline
							
Parking Birmingham~\cite{DBLP:conf/smartct/StolfiA017}	&	35,718	&		4		\\ \hline
Electric consumption~\cite{electricCon}	&	2,075,259	&		9	\\ \hline
Credit card clients~\cite{DBLP:journals/eswa/YehL09a}	&	30,000	&		23		\\ \hline
KDDCup99~\cite{KDD99cup}	&	494,021	&	42		\\ \hline
\end{tabular}
\end{adjustbox}
	\label{Table:dataset}
\end{table}

\vspace{-0.2in}

\subsection{Experimental Setting}

We perform all the experiments and comparisons on both privacy parameters ($\epsilon$ and $\gamma$) and anomaly detection parameters (the Threshold \textbf{T} and the sensitivity $\Delta q$). DPOAD is benchmarked with the Laplace mechanism and the recent Pain-Free mechanism for differential privacy.\footnote{\cite{DBLP:conf/cscml/BittnerSW18, bohler2017privacy,DBLP:conf/pkdd/OkadaFS15,DBLP:conf/ccs/AsifPV19, DBLP:conf/coinco/AsifTVSA16} can only disclose analysis results, i.e., detected anomalies (rather than data). Thus, they are incomparable with DPOAD.} We apply the original anomaly detection algorithm on all the datasets to obtain the baseline results in all the experiments for evaluating the accuracy of DPOAD and benchmarks. 




\begin{figure*}[ht]
	\centering
	\subfigure[Parking]{
		\includegraphics[angle=0, width=0.22\linewidth]{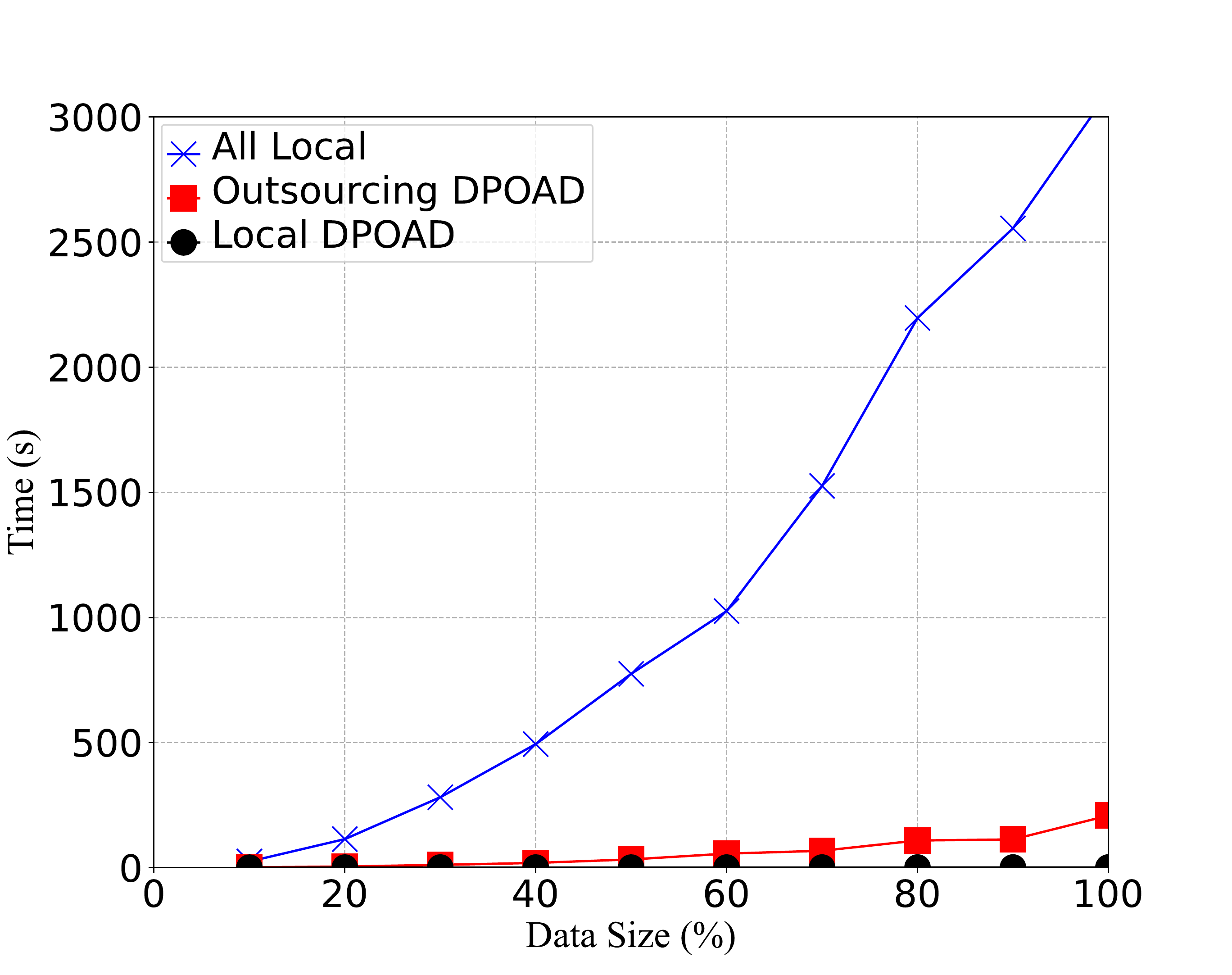}
		\label{fig:1_time} }
		\hspace{-0.22in}
	\subfigure[Electric Consumption]{
		\includegraphics[angle=0, width=0.22\linewidth]{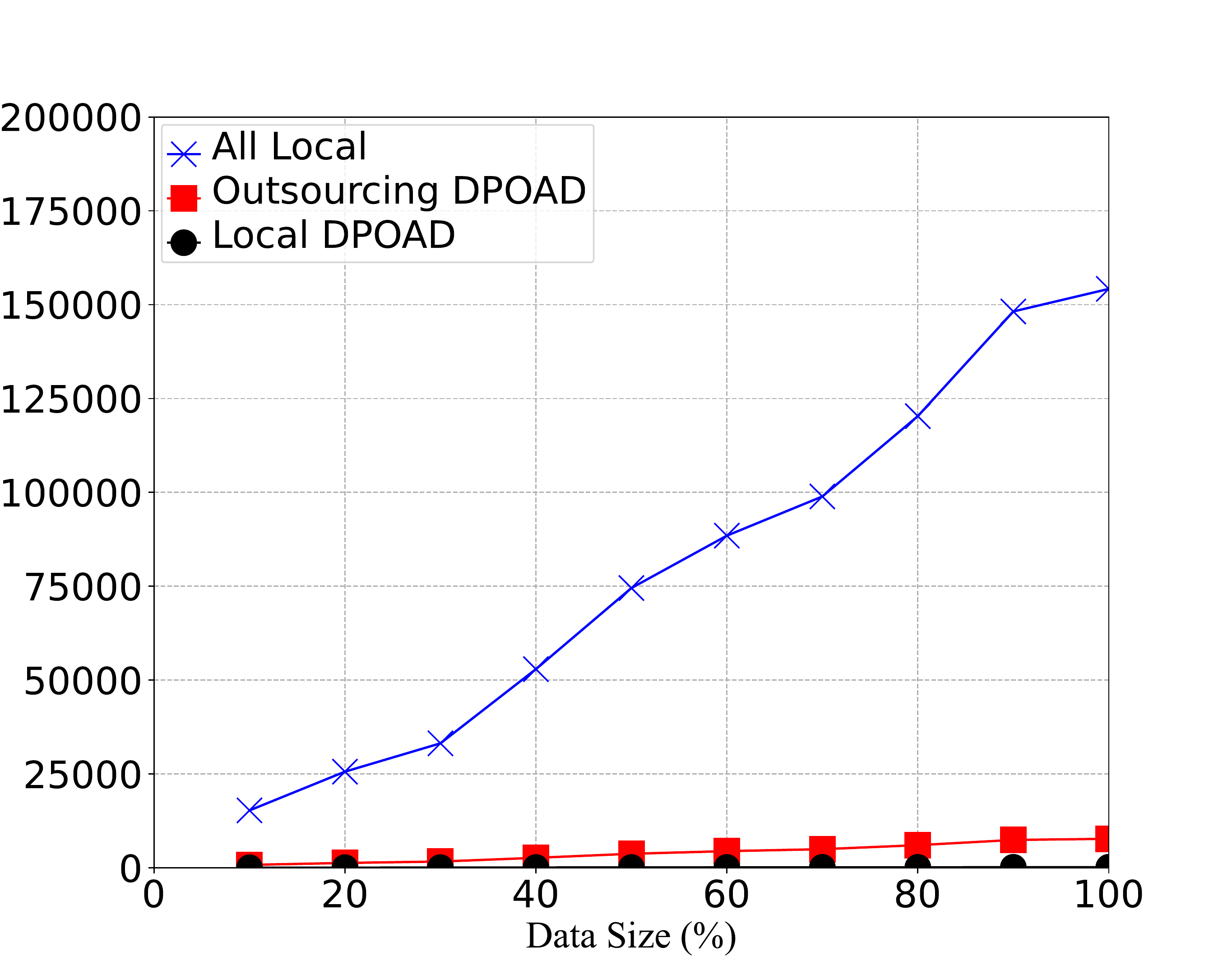}
		\label{fig:1_gam_Re}}
		\hspace{-0.22in}
	\subfigure[Credit Card]{
		\includegraphics[angle=0, width=0.22\linewidth]{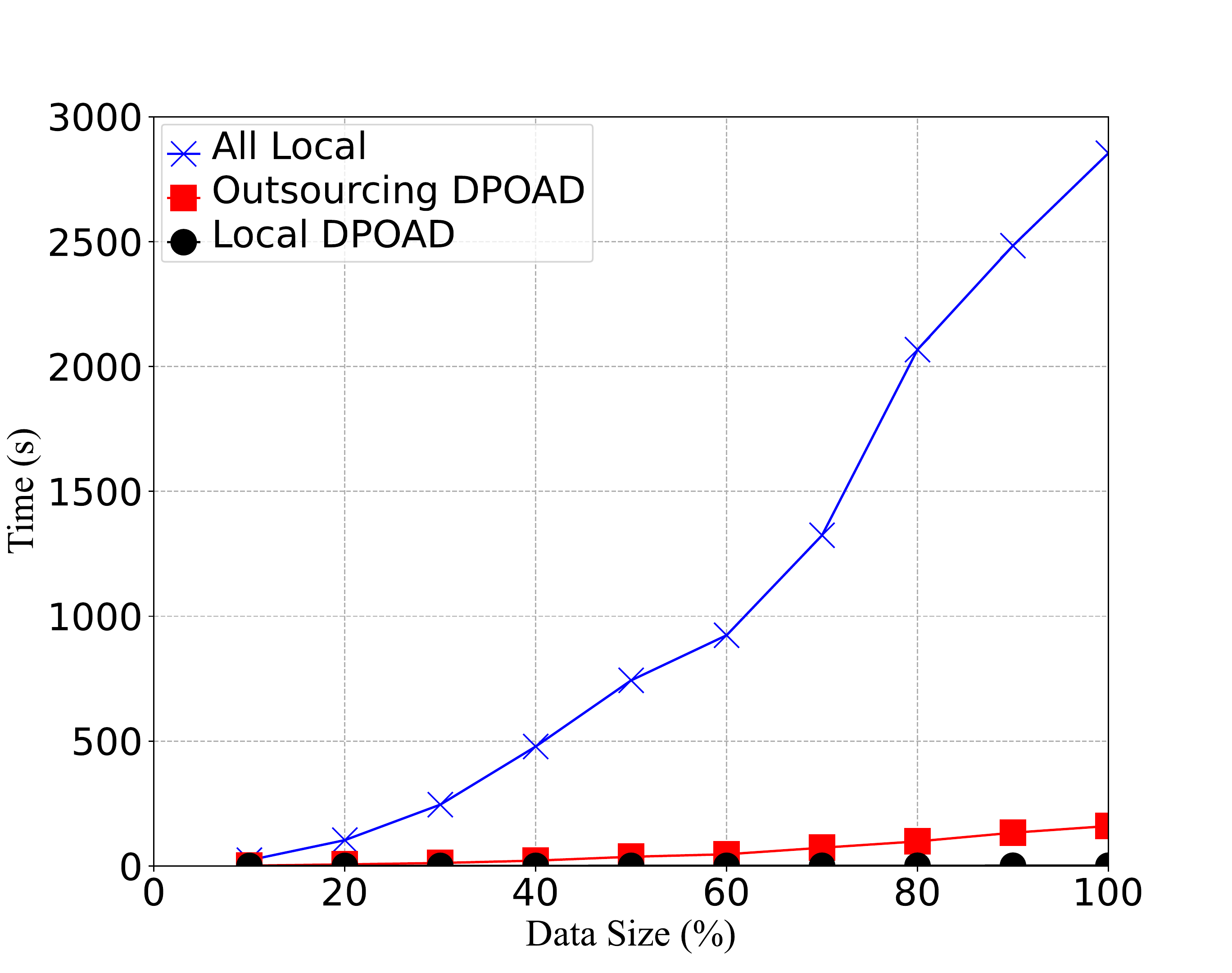}
		\label{fig:3_time} }
		\hspace{-0.22in}
	\subfigure[KDDCup99]{
		\includegraphics[angle=0, width=0.22\linewidth]{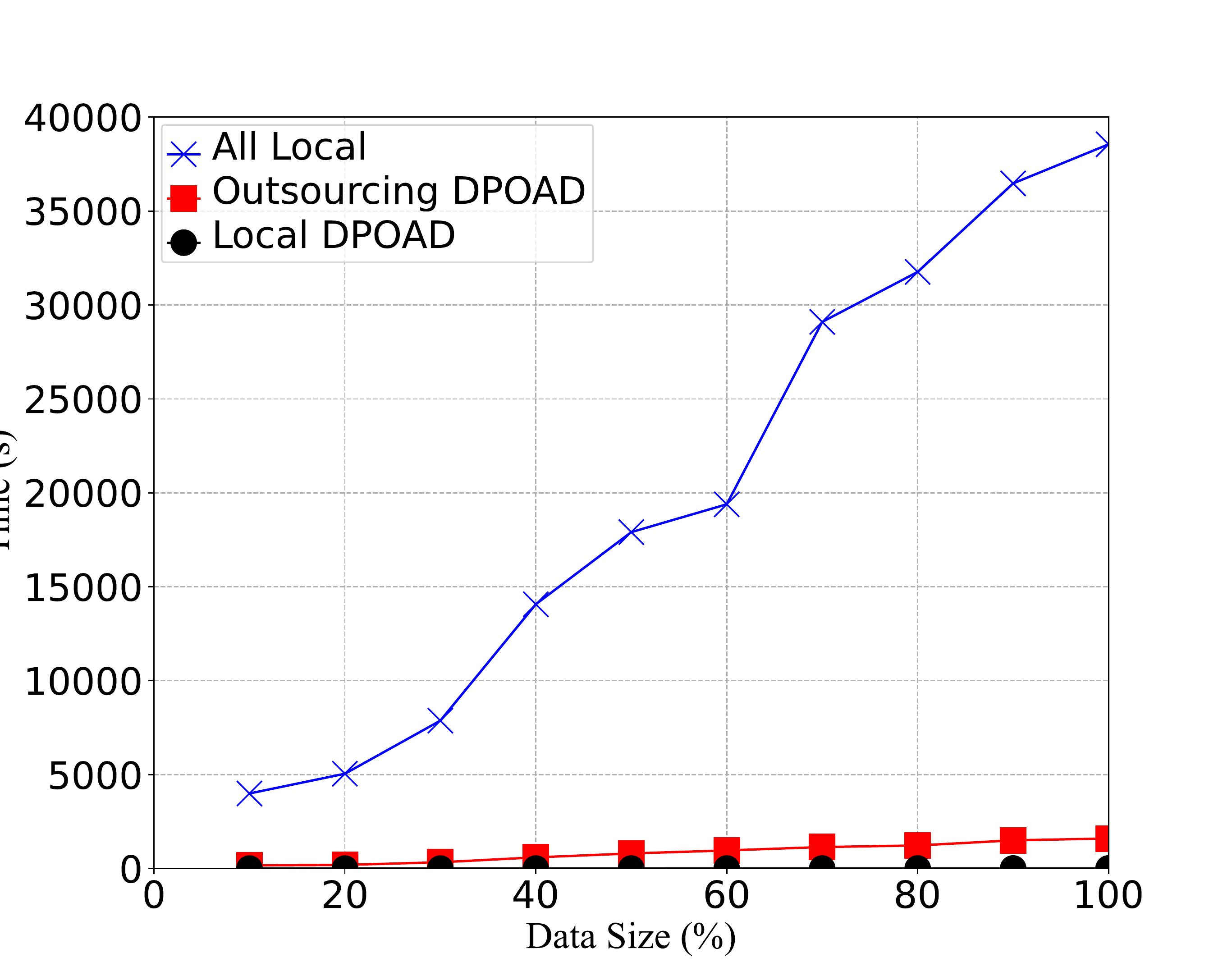}
		\label{fig:2_gam_Re}}
		\hspace{-0.22in}
	\caption{Time Performance Comparison of Local Setting and DPOAD (Outsourcing).}
	\label{figure:time}
\end{figure*}
\subhead{Definition of Anomalies}
In all the experiments, the MSSP runs anomaly detection based on Kolmogorov-Smirnov (KS) test~\cite{AnomalyMichael}. It is widely used in anomaly detection~\cite{moustafa2016evaluation,justel1997multivariate} as a reliable score-based anomaly detection. Specifically, it defines and compares the distribution of the training and testing sets. The randomly sampled training data can be transformed to a uniform distribution. Then, we apply the same transformation as defined in training data to test if the testing data points fall in the same uniform distribution on the $p$-dimensional hybercube. Any data point falls outside of the distribution would be considered as anomalies in this case. In general, the KS test evaluates the null hypotheses that the class labels predicted by a classifier are no better than random~\cite{BRADLEY2013470}.

\subhead{Privacy Parameters} 
In the first set of our experiments, we examine the accuracy of DPOAD on varying privacy parameters $\epsilon$ and $\gamma$, using two well-known accuracy metrics: precision and recall. In this group of experiments, we set the anomaly detection threshold as 0.9 and the number of analysis iterations as 5, respectively. 

In the initialization, we first sample the sensitivity $\Delta q$ from the uniform distribution defined over all possible values of a dataset. In the next iteration, we then update the underlying distribution using the anomaly scores computed by an MSSP and sample the new sensitivity with updated distribution. 



\subhead{Anomaly Detection Parameters}
In the second set of our experiments, we investigate how DPOAD behaves under different detection thresholds \textbf{T} and iterations from an MSSP. The threshold is the parameter that shows the abnormal rate in a dataset, which would affect the utility. Since the sample sensitivity $\Delta q$ in each iteration from the MSSP should be updated, the corresponding noise would then be different in each iteration. We run all experiments in the second set by fixing the privacy budget ($\epsilon=1$ and $\gamma=0.05$) and tuning anomaly detection parameters. 

\subhead{Platform} 
The experiments on the cloud (outsourced computation in DPOAD) were performed on the NSF Chameleon Cluster with Intel(R) Xeon(R)Gold 6126 2.60GHz CPUs and 192G RAM. The experiments under local setting were performed on an HP PC with Intel Core i7-7700 CPU 3.60GHz and 32G RAM. 
\subsection{Computation Advantage}
 The first group of experiments are conducted to evaluate the runtime of the algorithms by comparing the local anomaly detection and DPOAD (local noise injection and outsourced computation). We run all the algorithms on different percents of the overall data (the same four datasets). Figure \ref{figure:time} shows that the local computation can be significantly reduced if outsourcing the data to the cloud. Since the anomaly detection in the cloud is completed by multiple computing nodes with parallel computation, the runtime by the cloud is also significantly reduced compared to performing local anomaly detection. 

\begin{figure*}[!tbh]
	\centering
	\subfigure[Precision vs $\epsilon$]{
		\includegraphics[angle=0, width=0.22\linewidth]{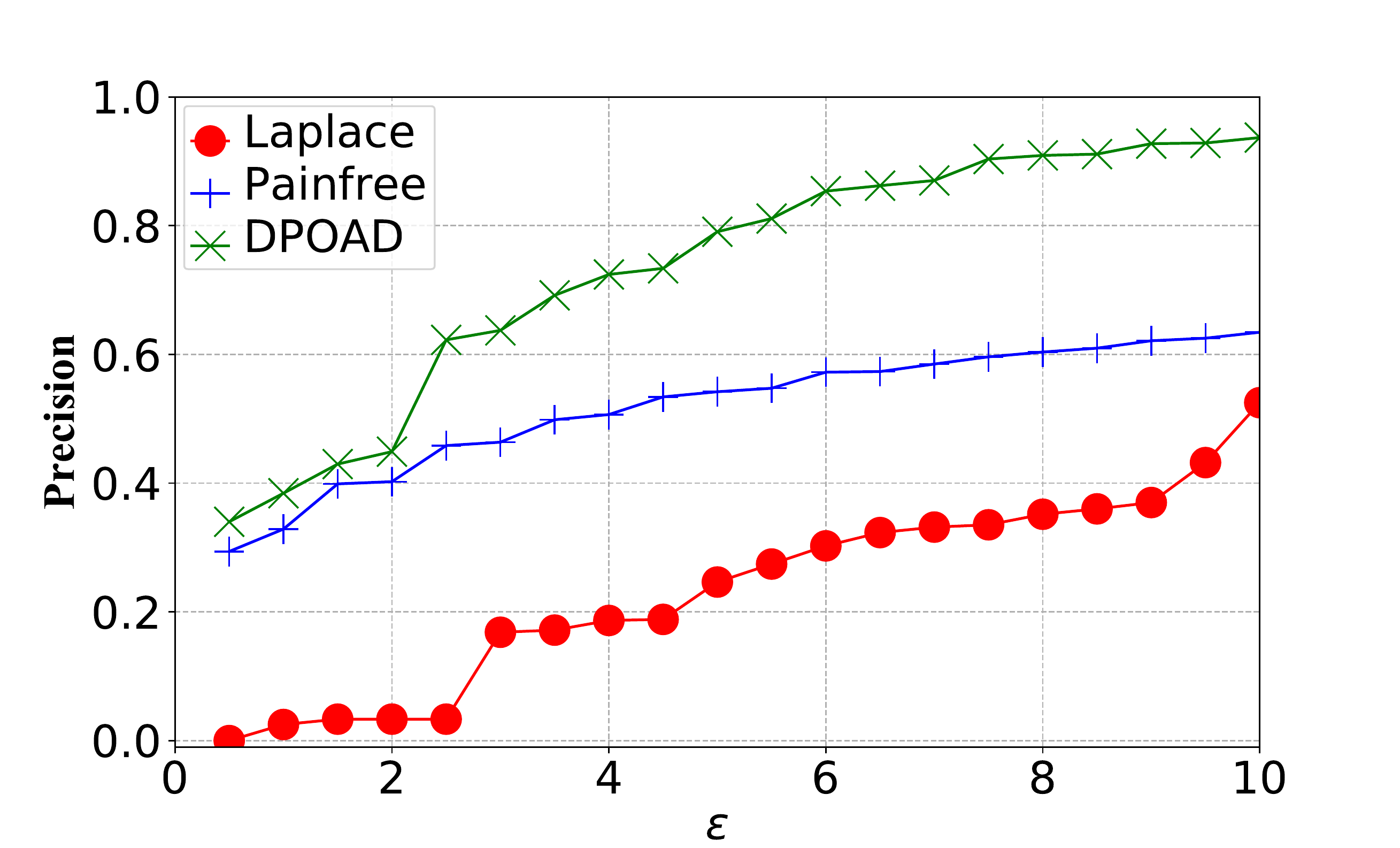}
		\label{fig:1Pre} }
		\hspace{-0.22in}
	\subfigure[Recall vs $\epsilon$]{
		\includegraphics[angle=0, width=0.22\linewidth]{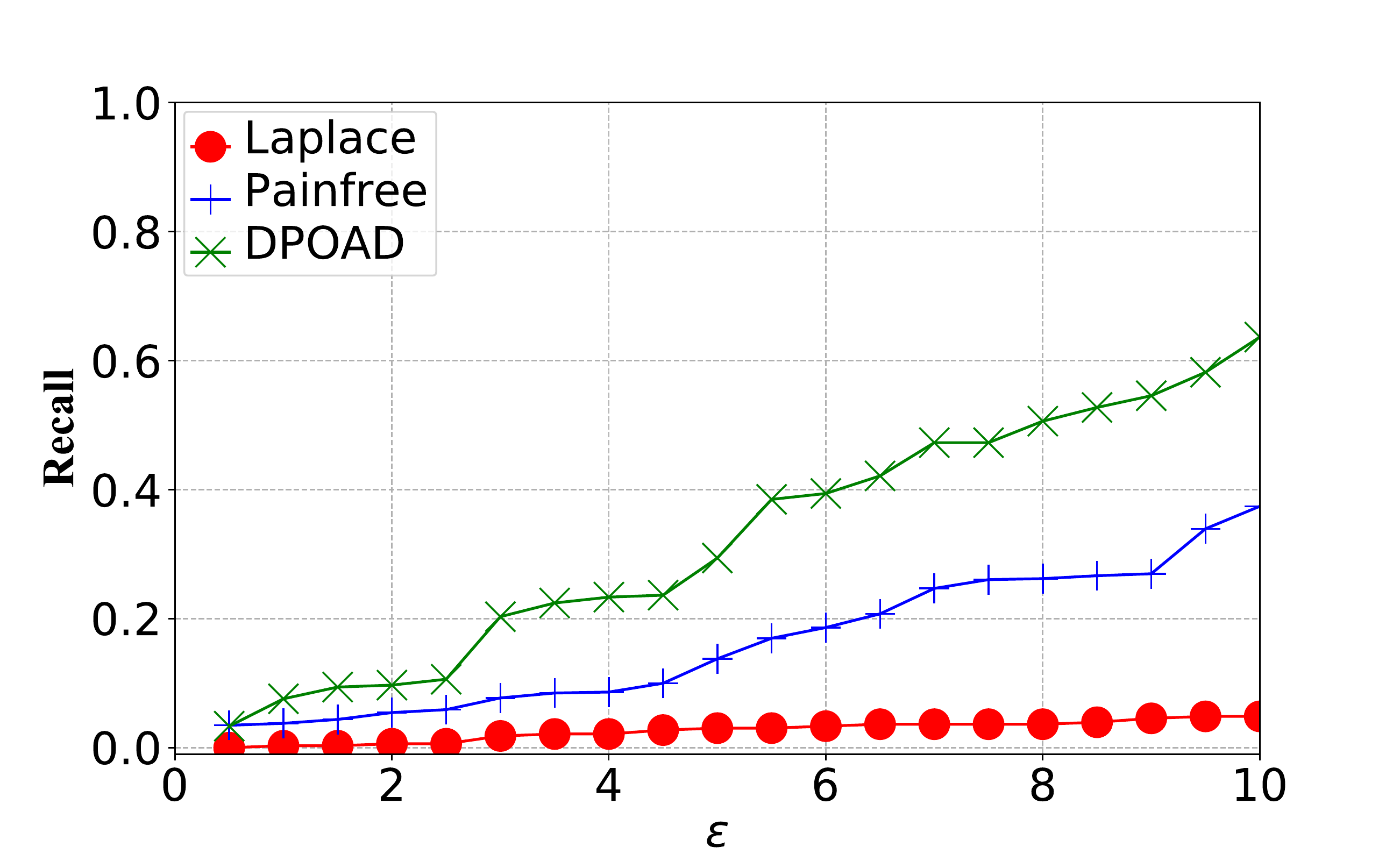}
		\label{fig:1Re}}
		\hspace{-0.22in}
	\subfigure[Precision vs $\epsilon$]{
		\includegraphics[angle=0, width=0.22\linewidth]{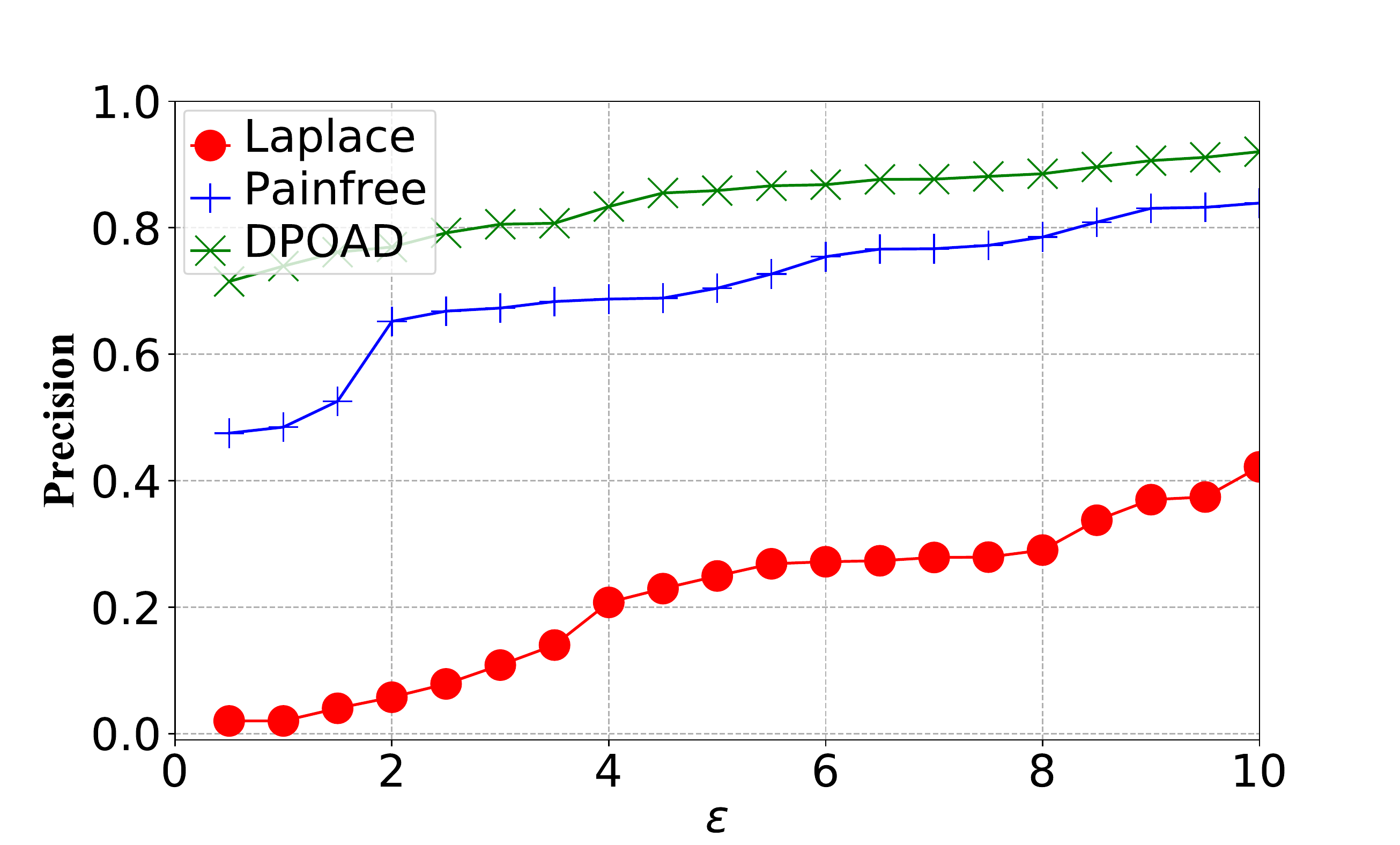}
		\label{fig:2Pre}}
		\hspace{-0.22in}
	\subfigure[Recall vs $\epsilon$]{
		\includegraphics[angle=0, width=0.22\linewidth]{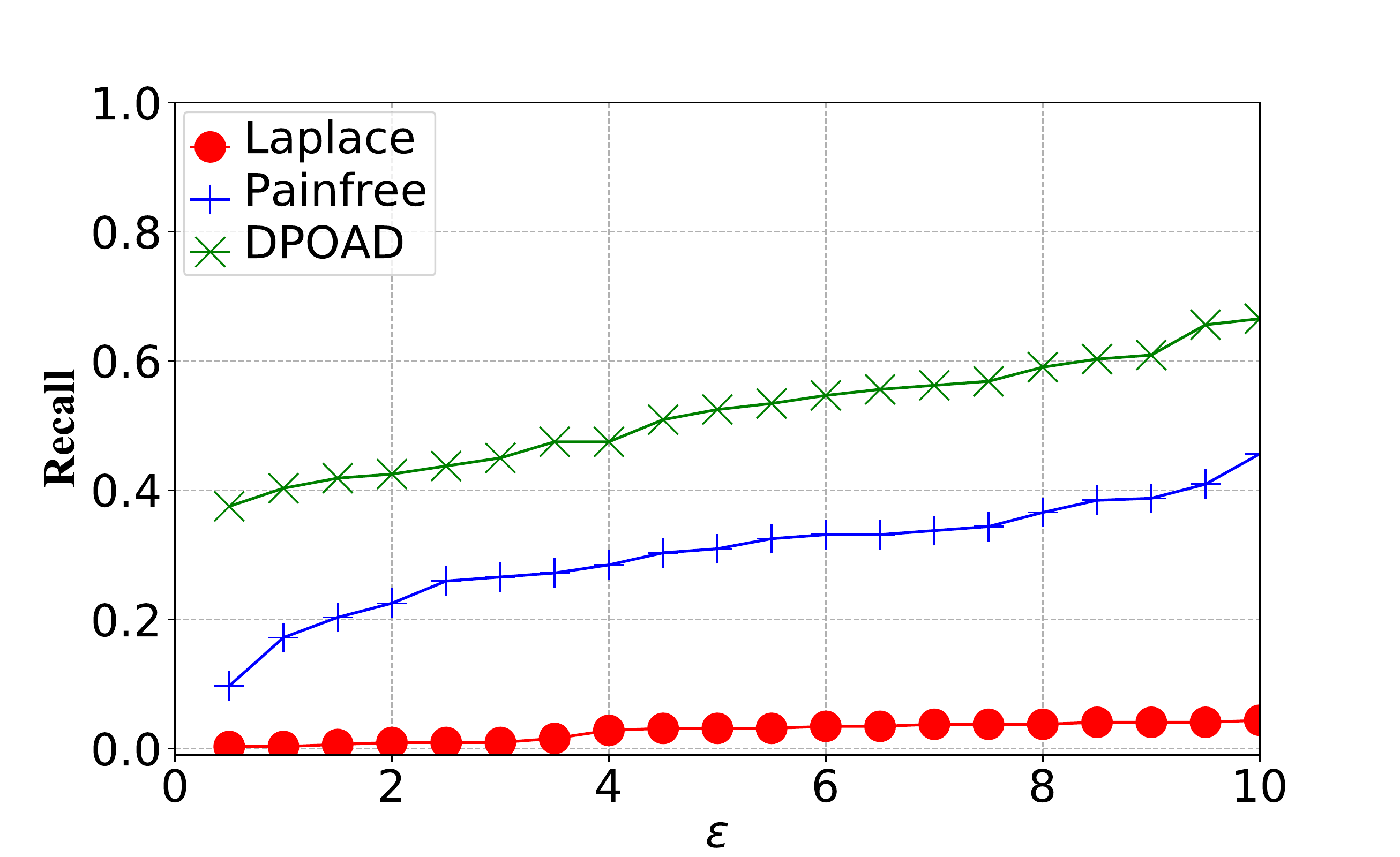}
		\label{fig:2Re}}
		\hspace{-0.22in}
	\subfigure[Precision vs $\epsilon$]{
		\includegraphics[angle=0, width=0.22\linewidth]{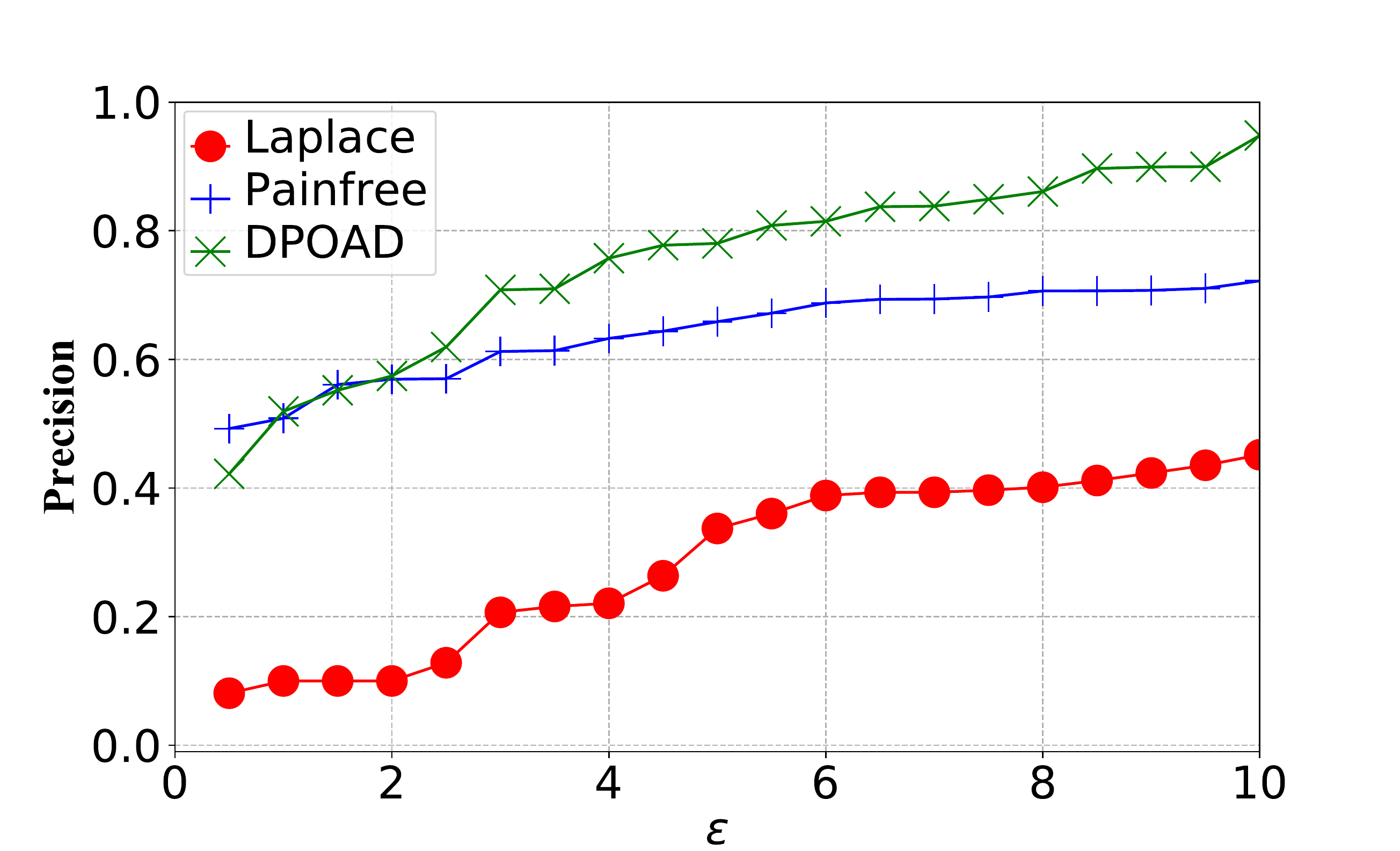}
		\label{fig:4Pre} }
		\hspace{-0.22in}
	\subfigure[Recall vs $\epsilon$]{
		\includegraphics[angle=0, width=0.22\linewidth]{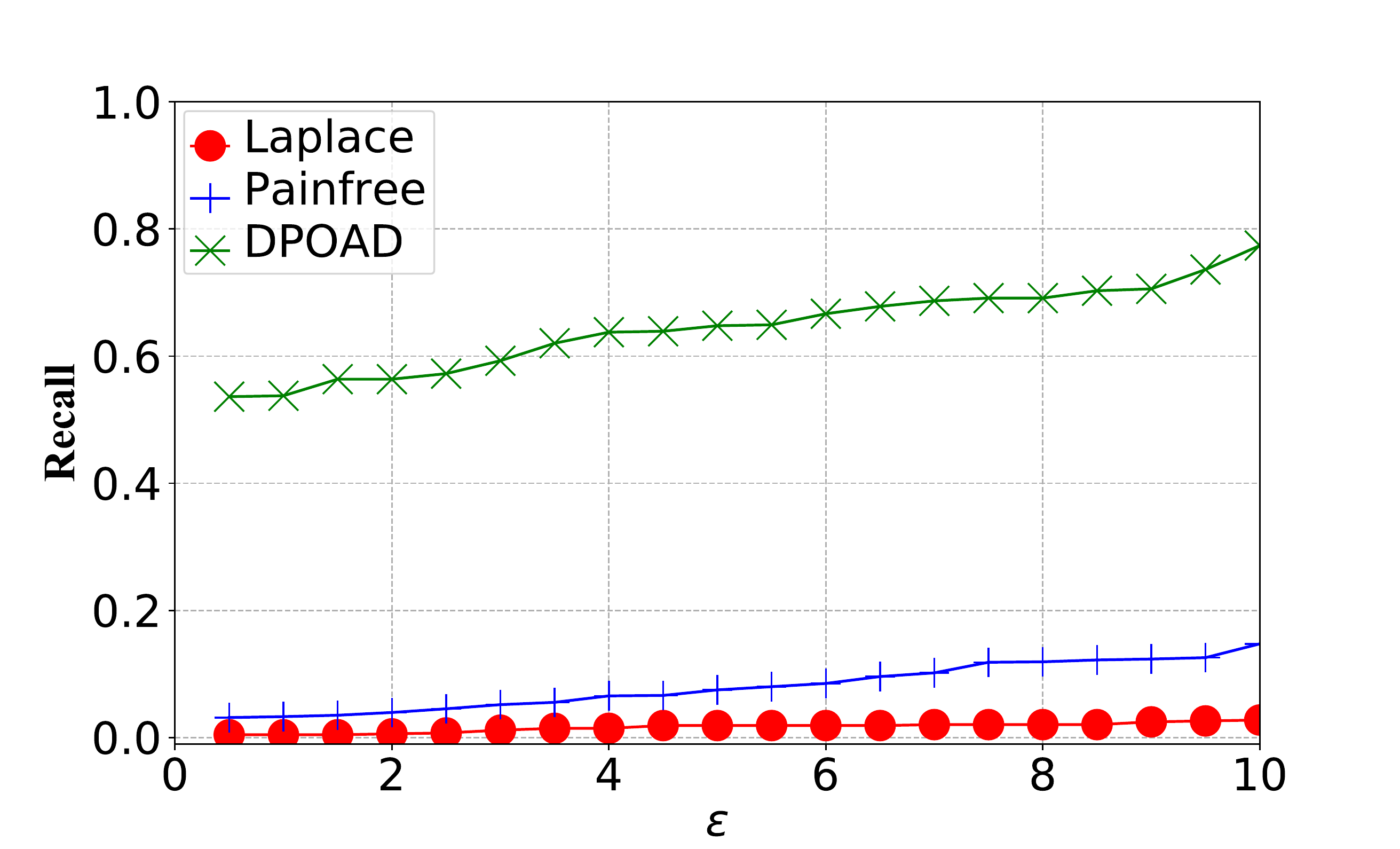}
		\label{fig:4Re}}
		\hspace{-0.22in}
	\subfigure[Precision vs $\epsilon$]{
		\includegraphics[angle=0, width=0.22\linewidth]{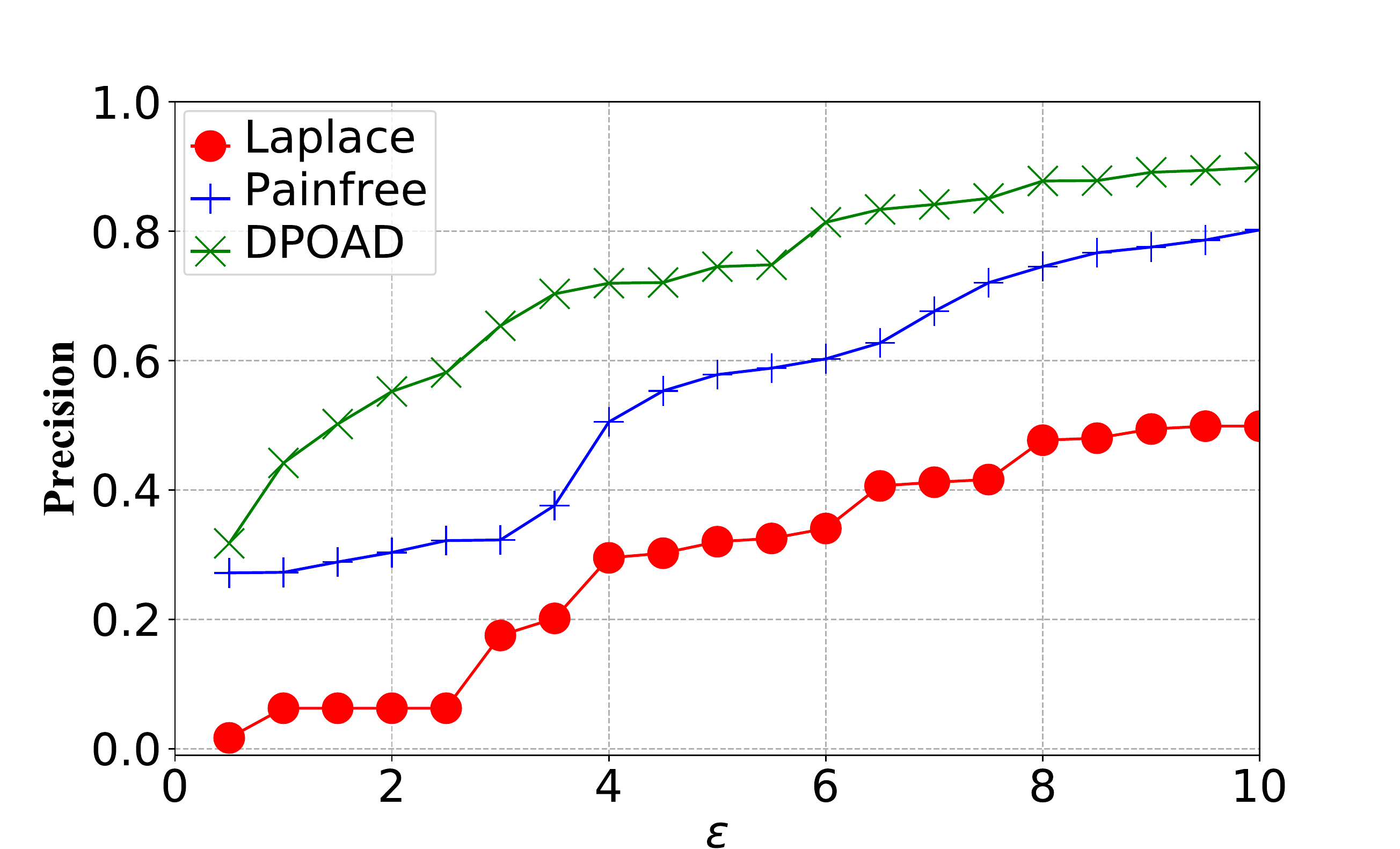}
		\label{fig:5Pre} }
	   \hspace{-0.22in}
	\subfigure[Recall vs $\epsilon$]{
		\includegraphics[angle=0, width=0.22\linewidth]{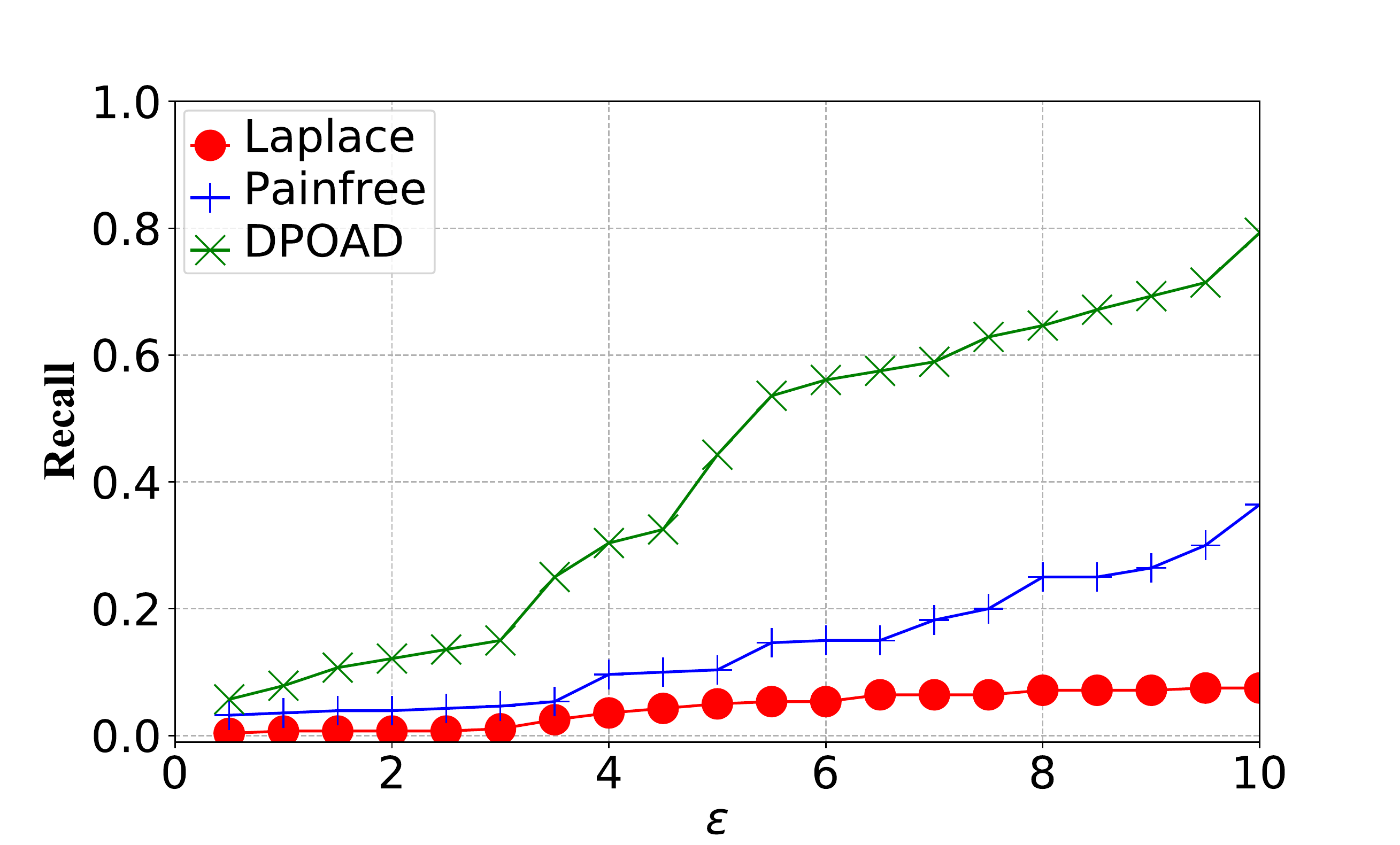}}
		\vspace{-0.15in}
	\caption{Precision and Recall vs $\epsilon$ on the parking dataset (a,b), electric consumption dataset (c,d), credit card clients dataset (e,f) and KDDCup99 dataset (g,h).}\vspace{-0.1in}
	\label{fig:eps}
\end{figure*}

\subsection{Privacy Parameters} \label{subsec:privacyresults}

We study the effectiveness of DPOAD by varying the privacy parameters, i.e., $\epsilon$ and $\gamma$ in all experiments is fixed to $0.1$ which regards to the accuracy of the PDF estimator in the learning phase. The precision and recall values shown in the figures are calculated based on the baseline results (the results are obtained without any injected noise). In the studied accuracy metrics, recall shows the sensitivity of the approaches, while precision demonstrates the ratio of correctly identified anomalies.


\subsubsection{Accuracy vs Privacy Budget $ \epsilon$}
For each dataset, as described in Table~\ref{Table:dataset}, we evaluate the accuracy for anomaly detection based on two metrics by varying $\epsilon$ ($\gamma$ is fixed as 0.3 in this set of experiments). Laplace mechanism has been used as the baseline in all the studies. Overall, the value of metrics increases when the privacy budget increases. 
Specifically, we have the following observation:

\begin{itemize}

    \item In Figure~\ref{fig:eps}, the DPOAD-based anomaly detection is significantly more accurate than Laplace and Pain-Free with the same privacy budget for all the datasets. The precision and recall results are much higher than the benchmarks. This is because the noise injected to anomalies by DPOAD is much smaller, and thus more anomalies can be accurately identified.
    

    
    \item Similar to the two other solutions, the accuracy of DPOAD-based anomaly detection increases as the privacy protection becomes weaker (larger $\epsilon$), and
    the accuracy also relies on the characteristics of the datasets. For example, compared to other datasets the records in both credit card clients are more homogeneous, which lead to slightly smaller recall values in Figure~\ref{fig:4Re}. It is worth noting that the recall values can be simply improved by applying better anomaly detection algorithms w.r.t. such type of datasets. 

\end{itemize}

\subhead{Summary} DPOAD achieves much better precision and recall results than two benchmarks on all the datasets. In both strong and weak privacy cases, DPOAD identifies anomalies more accurately.

\begin{figure*}[!tbh]
	\centering
	\subfigure[Precision vs Threshold]{
		\includegraphics[angle=0, width=0.22\linewidth]{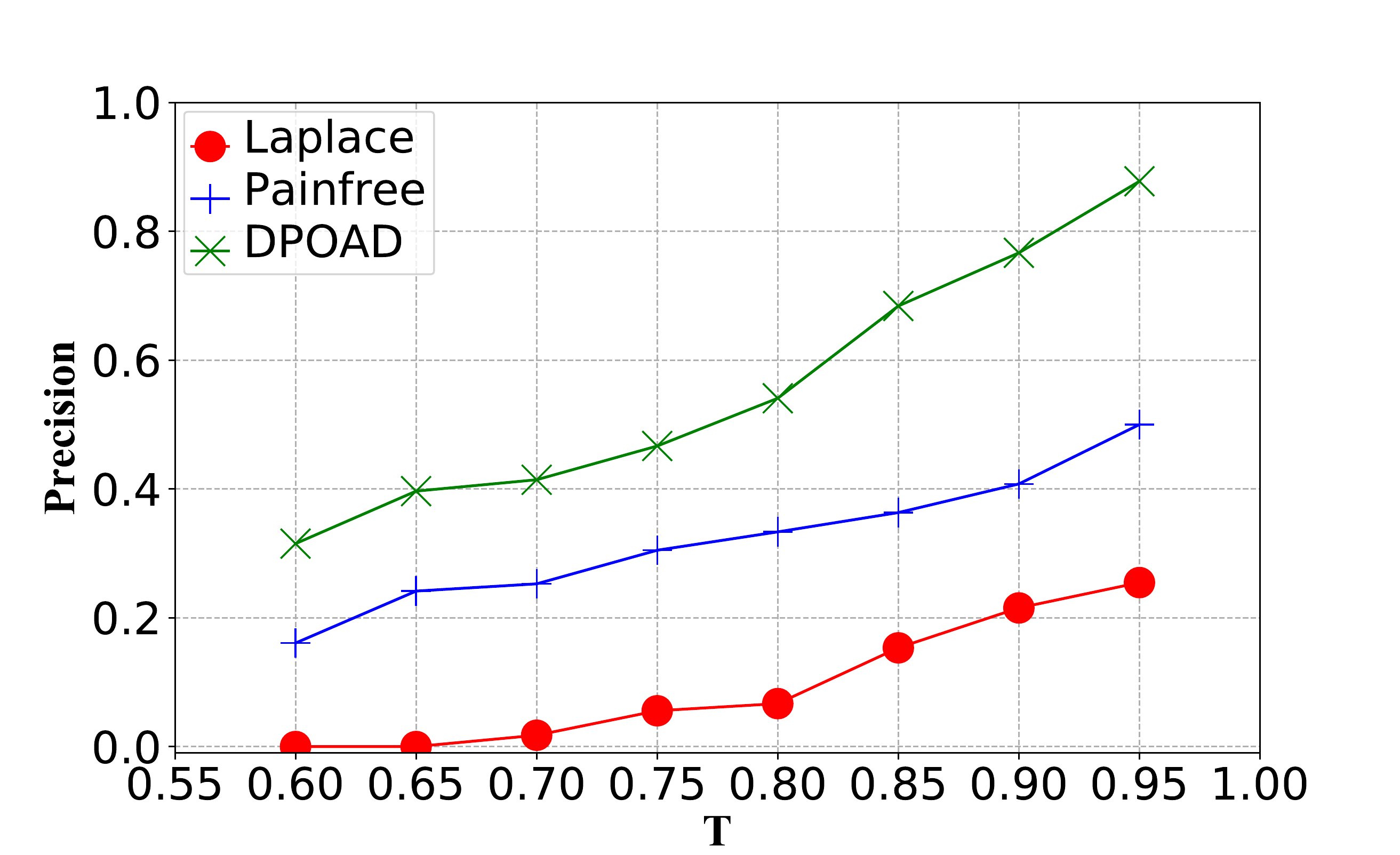}
		\label{fig:1_T_Pre} }
		\hspace{-0.22in}
	\subfigure[Recall vs Threshold]{
		\includegraphics[angle=0, width=0.22\linewidth]{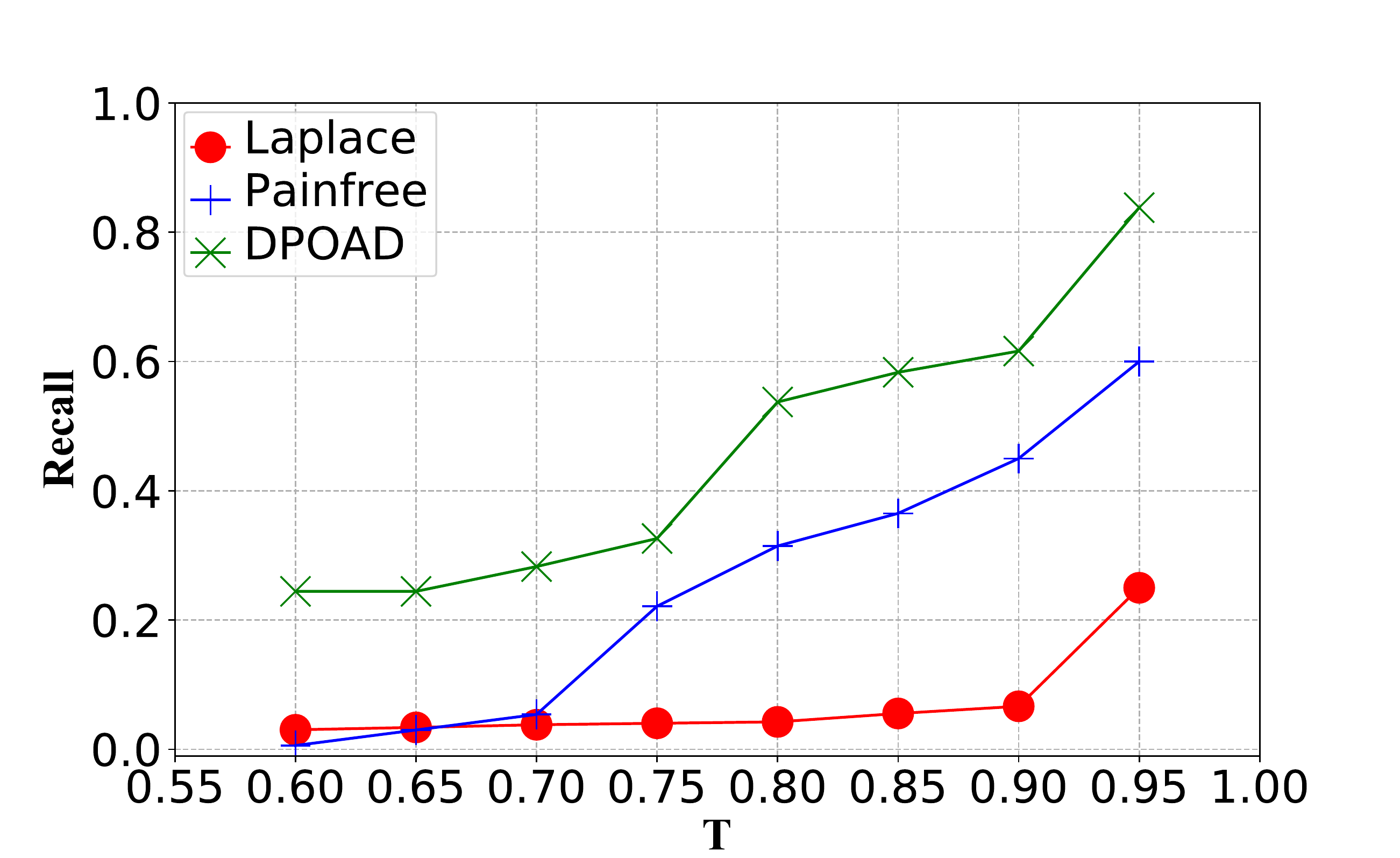}
		\label{fig:1_T_Re}}
		\hspace{-0.22in}
	\subfigure[Precision vs Threshold]{
		\includegraphics[angle=0, width=0.22\linewidth]{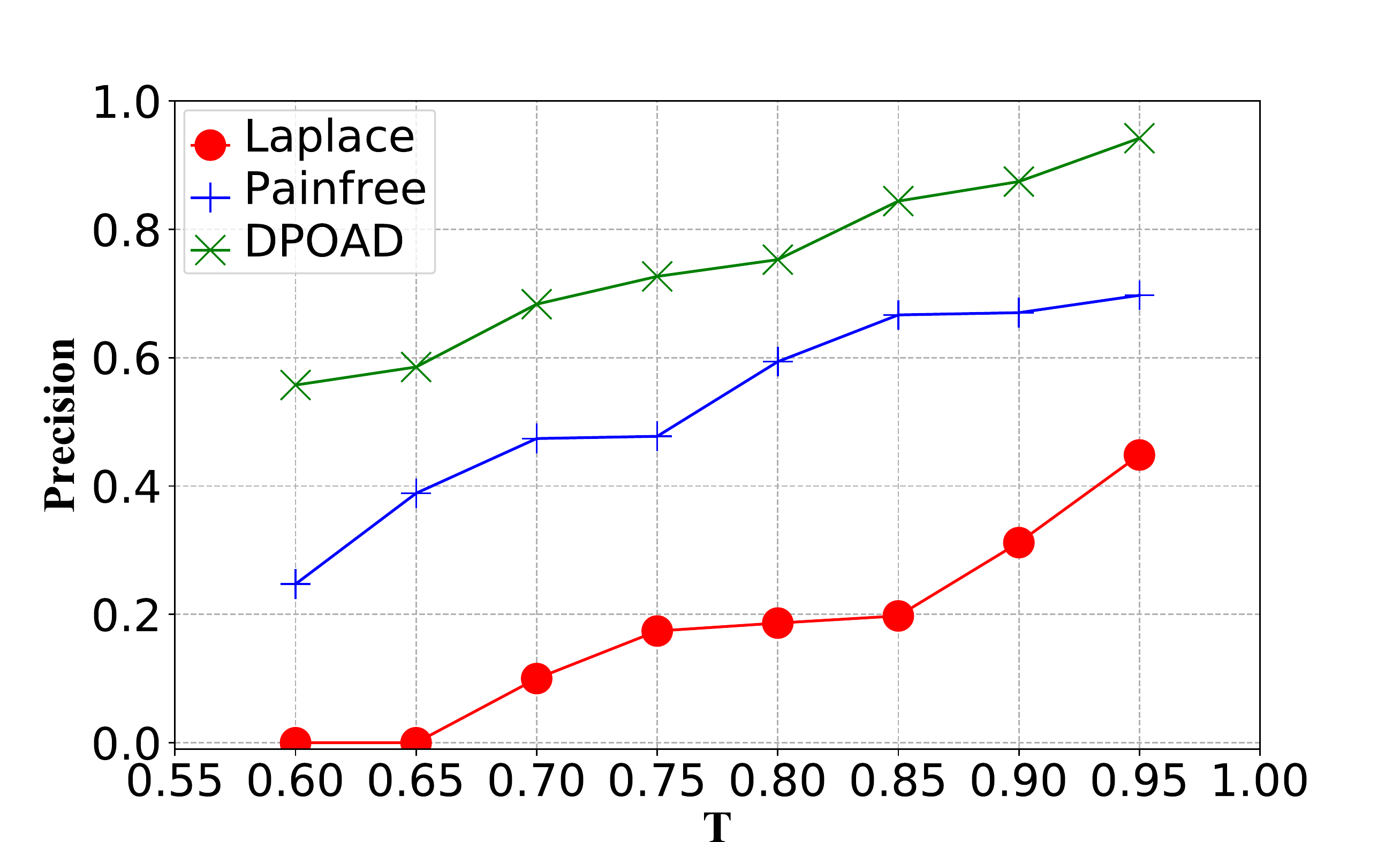}
		\label{fig:2_T_Pre} }
		\hspace{-0.22in}
	\subfigure[Recall vs Threshold]{
		\includegraphics[angle=0, width=0.22\linewidth]{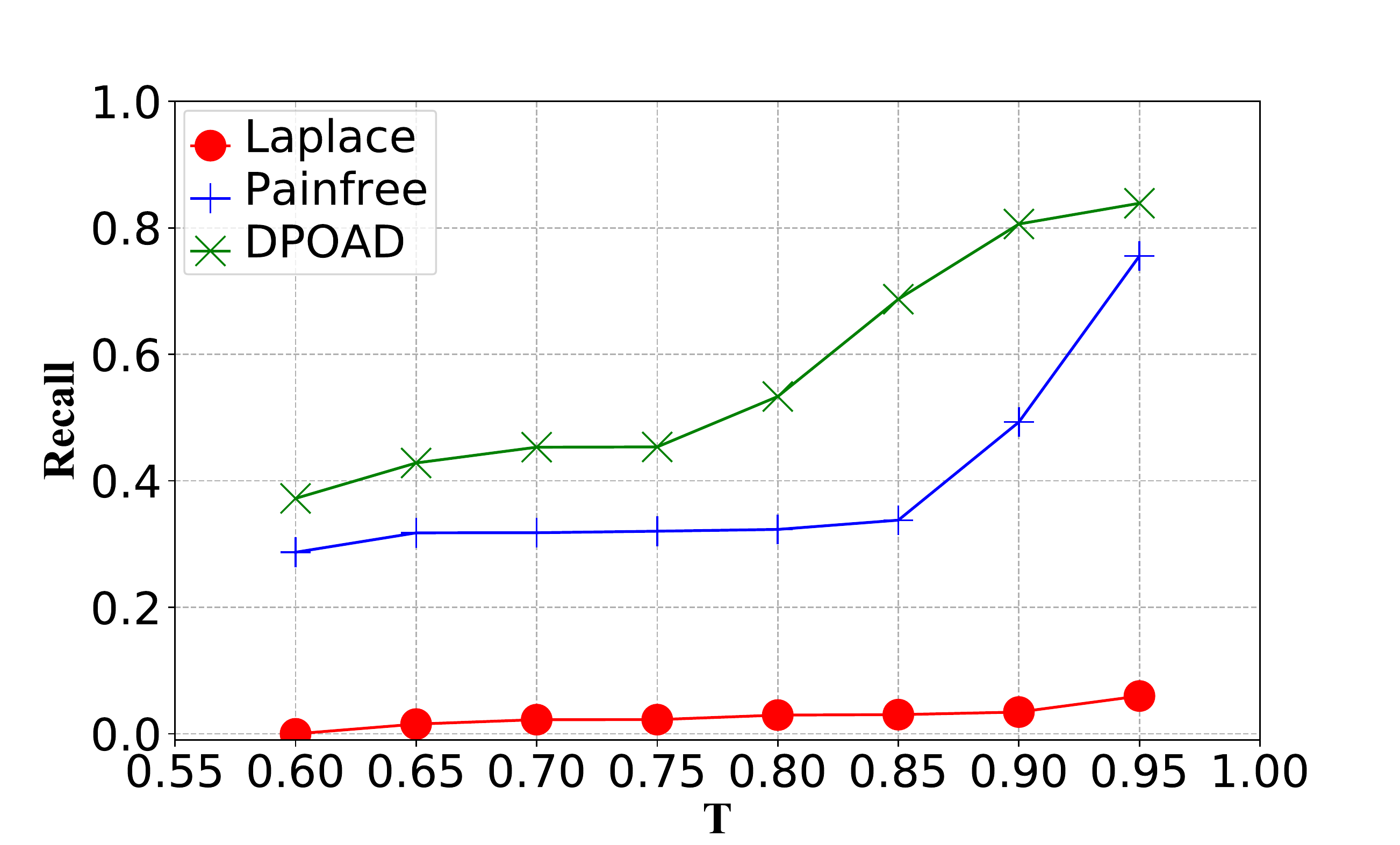}
		\label{fig:2_T_Re}}
		\hspace{-0.22in}
	\subfigure[Precision vs Threshold]{
		\includegraphics[angle=0, width=0.22\linewidth]{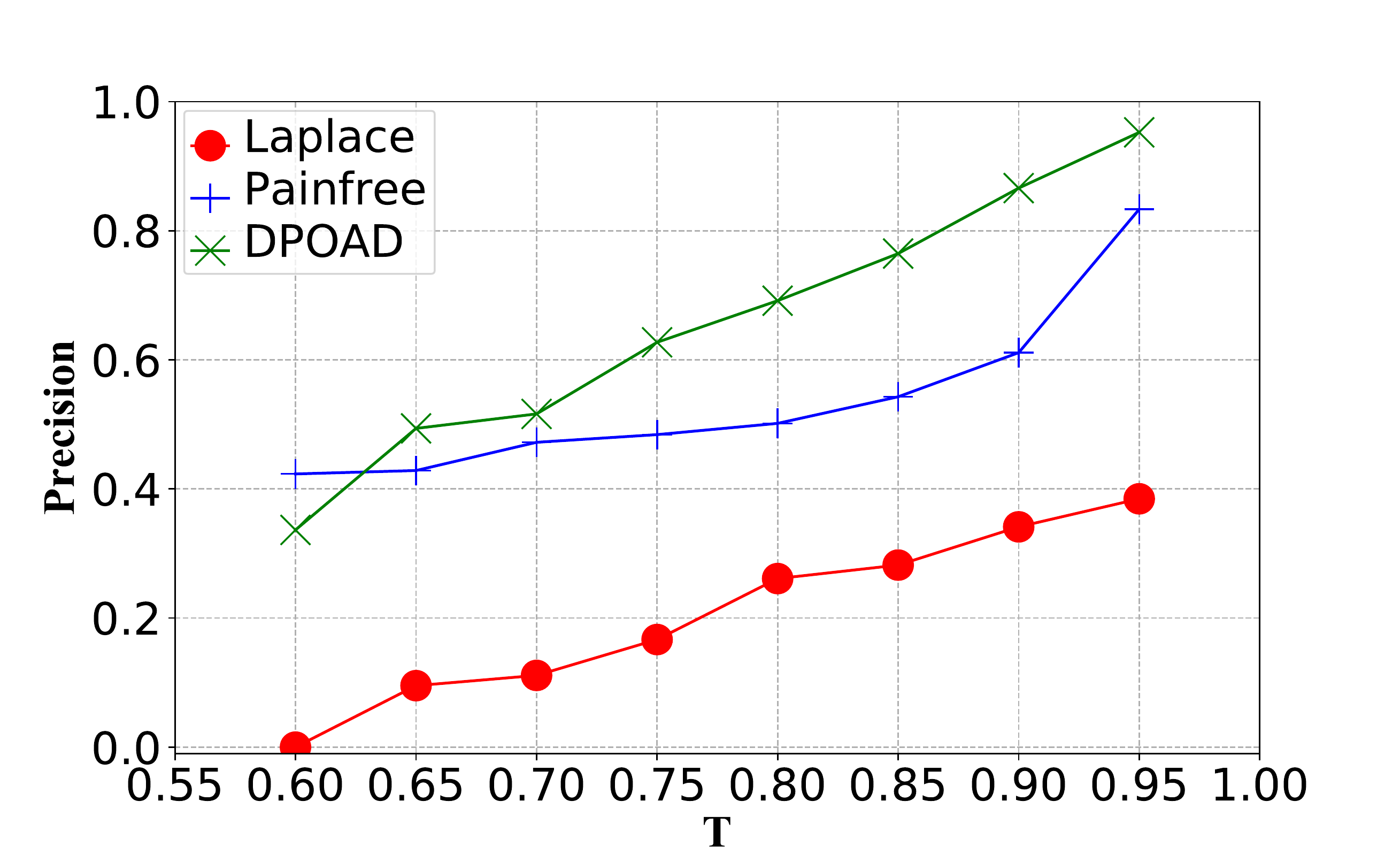}
		\label{fig:4_T_Pre} }
		\hspace{-0.22in}
	\subfigure[Recall vs Threshold]{
		\includegraphics[angle=0, width=0.22\linewidth]{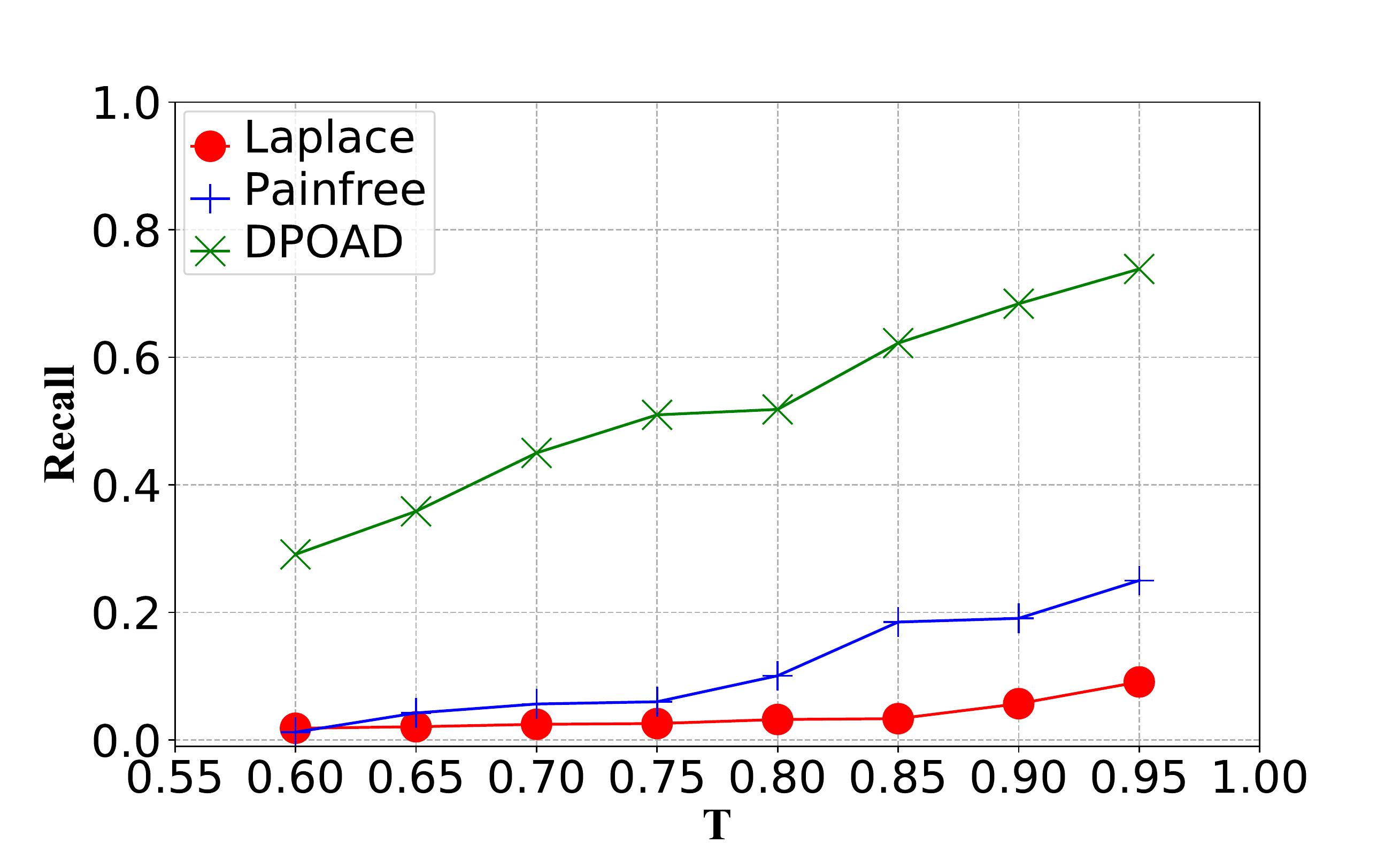}
		\label{fig:4_T_Re}}
		\hspace{-0.22in}
	\subfigure[Precision vs Threshold]{
		\includegraphics[angle=0, width=0.22\linewidth]{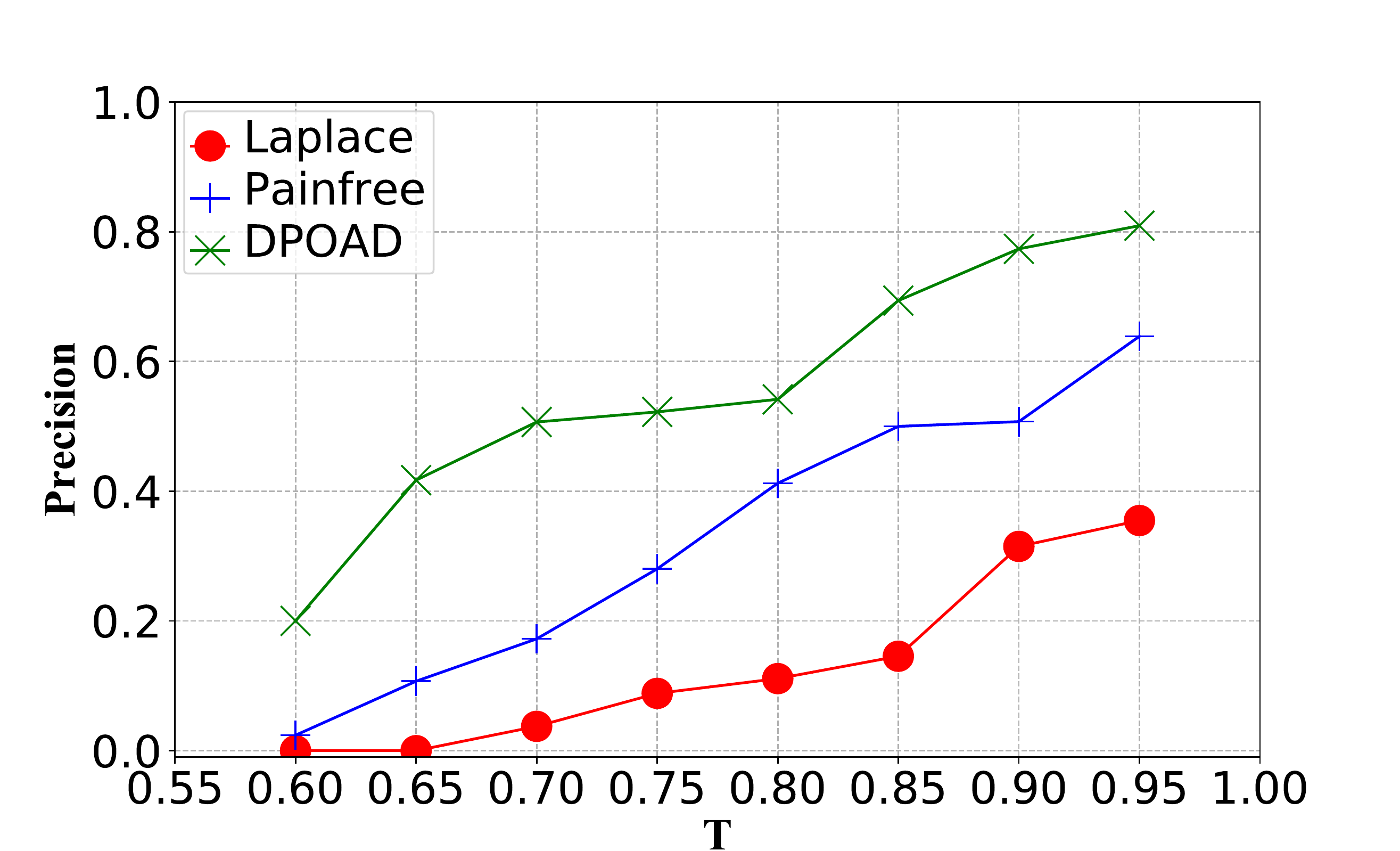}
		\label{fig:5_T_Pre} }
		\hspace{-0.22in}
	\subfigure[Recall vs Threshold]{
		\includegraphics[angle=0, width=0.22\linewidth]{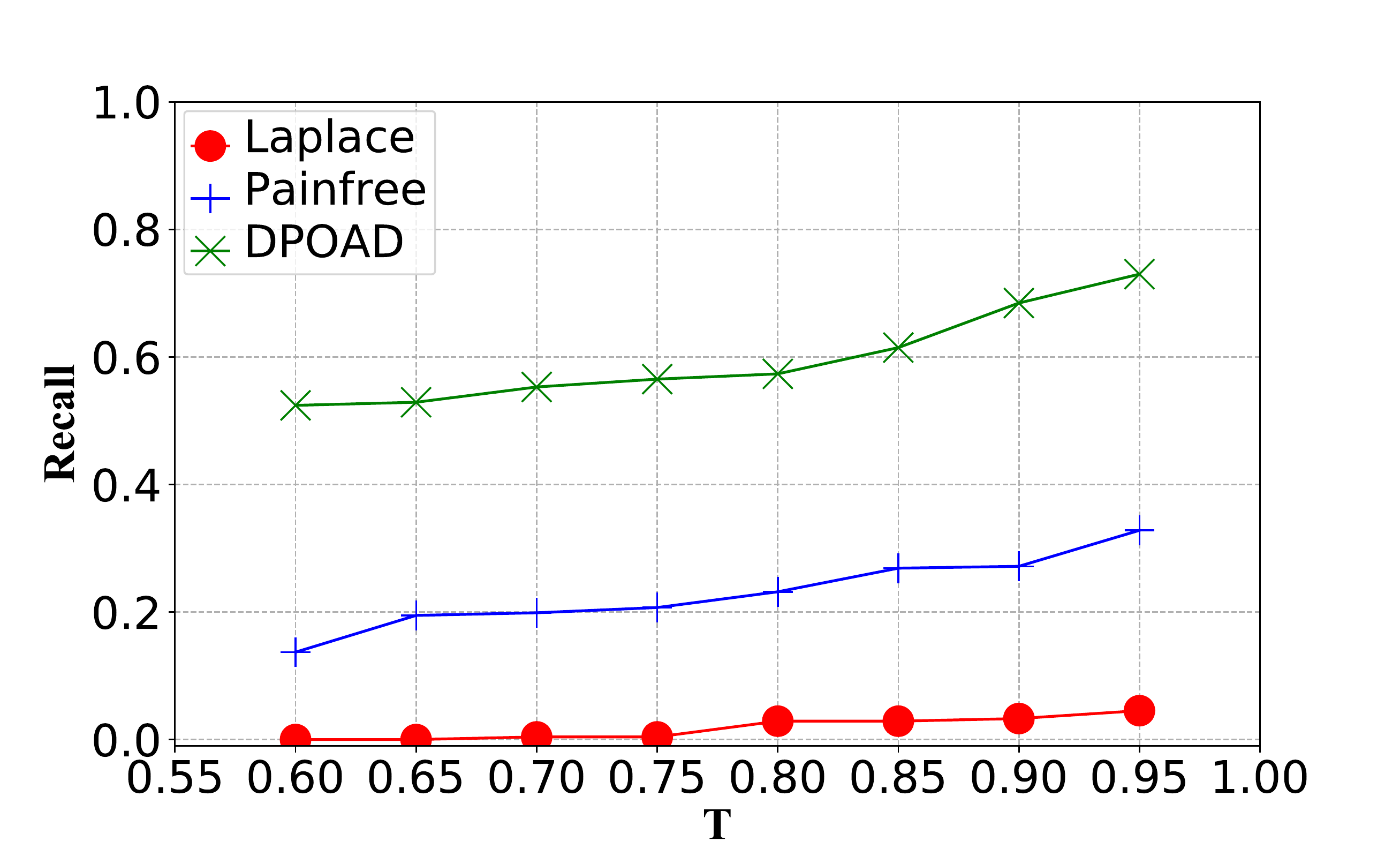}
		\label{fig:5_T_Re}}
	\caption{
	Precision and Recall vs Detection Threshold (\textbf{T}) on the parking dataset (a,b), electric consumption dataset (c,d), credit card clients dataset (e,f) and KDDCup99 dataset (g,h).
	}
	\label{fig:thre}
\end{figure*}

\subsection{Anomaly Detection Parameters} \label{subsec:anamalyP}

This group of experiments is conducted to provide insights about the anomaly detection parameters in the six datasets for DPOAD and two benchmarks. In particular, Figure~\ref{fig:thre} depicts the impact of the chosen threshold on anomaly detection, while Table~\ref{tab:rounds} shows the impact of the number of analysis iterations. 
 
\subsubsection{Accuracy vs Detection Threshold} \label{subsubsec:threshold}

Unlike the anomaly detection, i.e., the ground truth for comparison is fixed, the threshold in this set of experiments also impacts the baseline solution. Thus, the recall and precision in all the experiments grow when the threshold increases. This confirms that DPOAD can apply minor noise in the explicit anomalies, i.e., threshold $\geq$ 0.95, while only applying privacy to normal data. Specifically, with a smaller pool of highly anomalous items (higher threshold), DPOAD has a even better precision and recall on all the datasets. On the contrary, the accuracy of Laplace cannot increase much as the threshold increases.





\begin{table*}[!tbh]
\centering 
\caption{Precision and Recall vs Analysis Iterations.}
\label{tab:rounds}
\scalebox{0.87}{
\begin{tabular}{cc}   
Parking Dataset & Electric Consumption Dataset \\ 
\begin{tabular}{|c|c|c|c|c|c|c|c|}
\hline
&Iteration &1&2&3&4&5&6\\
\hline
\multirow{3}{*}{Precision} & Laplace & 0.00 &
0.05 &
0.13 &
0.18 &
0.24 &
0.25 \\
\cline{2-8}
& Painfree & 0.03 &
0.04 &
0.23 &
0.28 &
0.38 &
0.41
\\
\cline{2-8}
& \textbf{DPOAD} & \textbf{0.20} &
\textbf{0.43} &
\textbf{0.68} &
\textbf{0.80} &
\textbf{0.82} &
\textbf{0.84}
\\
\hline
\hline
\multirow{3}{*}{Recall} & Laplace & 0.00 &
0.00 &
0.01 &
0.07 &
0.08 &
0.11\\
\cline{2-8}
& Painfree & 0.00 &
0.11 &
0.13 &
0.19 &
0.23 &
0.27\\
	
\cline{2-8}
& \textbf{DPOAD} & \textbf{0.10}
& \textbf{0.26} &
\textbf{0.33} &
\textbf{0.45} &
\textbf{0.56} &
\textbf{0.63}\\
\hline
\end{tabular}
&  
\begin{tabular}{|c|c|c|c|c|c|c|c|}
\hline
&Iteration &1&2&3&4&5&6\\
\hline
\multirow{3}{*}{Precision} & Laplace & 0.00 &
0.05 &
0.08 &
0.17 &
0.19 &
0.20 \\
\cline{2-8}
& Painfree & 0.11 &
0.25 &
0.30 &
0.48 &
0.52 &
0.58 \\
	
\cline{2-8}
& \textbf{DPOAD} &
\textbf{0.27} &
\textbf{0.37} &
\textbf{0.47} &
\textbf{0.63} &
\textbf{0.77} &
\textbf{0.83} 
\\
\hline
\hline
\multirow{3}{*}{Recall} & Laplace & 0.00 &
0.01 &
0.01 &
0.01 &
0.03 &
0.03 \\
\cline{2-8}
& Painfree & 0.01 &
0.02 &
0.13 &
0.23 &
0.35 &
0.41\\
	
\cline{2-8}
& \textbf{DPOAD} & \textbf{0.25} &
\textbf{0.42} &
\textbf{0.50} &
\textbf{0.55} &
\textbf{0.58} &
\textbf{0.67}\\
\hline
\end{tabular}
\end{tabular}}
\end{table*}

\begin{table*}[!tbh]
\centering 
\scalebox{0.87}{
\begin{tabular}{cc}
Credit Card Clients Dataset & KDDCup99 Dataset \\ 
\begin{tabular}{|c|c|c|c|c|c|c|c|}
\hline
&Iteration &1&2&3&4&5&6\\
\hline
\multirow{3}{*}{Precision} & Laplace & 0.20 &
0.22 &
0.24 &
0.25 & 
0.28 &
0.32 \\
\cline{2-8}
& Painfree & 0.23 &
0.34 &
0.36 &
0.37 &
0.40 &
0.41 
\\
\cline{2-8}
& \textbf{DPOAD} & \textbf{0.38}&
\textbf{0.44} &
\textbf{0.50} &
\textbf{0.62} &
\textbf{0.74} &
\textbf{0.81} 
\\
\hline
\hline
\multirow{3}{*}{Recall} & Laplace & 0.00 &
0.00 &
0.02 &
0.04 &
0.05 &
0.05 \\
\cline{2-8}
& Painfree & 0.01 &
0.02 &
0.06 &
0.08 &
0.08 &
0.10 
\\
\cline{2-8}
& \textbf{DPOAD} & \textbf{0.21}&
\textbf{0.29} &
\textbf{0.33} &
\textbf{0.40} &
\textbf{0.57} &
\textbf{0.60} \\
\hline
\end{tabular}
&  
\begin{tabular}{|c|c|c|c|c|c|c|c|}
\hline
&Iteration &1&2&3&4&5&6\\
\hline
\multirow{3}{*}{Precision} & Laplace & 0.00 &
0.08 &
0.15 &
0.22 &
0.27 &
0.37 \\
\cline{2-8}
& Painfree & 0.33 &
0.39 &
0.43 &
0.47 &
0.56 &
0.64 
	\\
	
\cline{2-8}
& \textbf{DPOAD} & \textbf{0.42} &
\textbf{0.52} &
\textbf{0.60} &
\textbf{0.71} &
\textbf{0.75} &
\textbf{0.81} 
\\
\hline
\hline
\multirow{3}{*}{Recall} & Laplace & 0.00 &
0.01 &
0.02 &
0.05 &
0.10 &
0.11 \\
\cline{2-8}
& Painfree & 0.04 &
0.05 &
0.06 &
0.06 &
0.10 &
0.17 
	\\
	
\cline{2-8}
& \textbf{DPOAD} & \textbf{0.36} &
\textbf{0.37} &
\textbf{0.43} &
\textbf{0.49} &
\textbf{0.55} &
\textbf{0.67}  \\
\hline
\end{tabular}
\end{tabular}
}
\end{table*}

\subsubsection{Accuracy vs Analysis Iteration} \label{subsubsec:round}

 In this set of experiments, we study the effectiveness of the analysis iterations on each dataset. Recall that the data owner releases the dataset collected between time $t_{i-1}\leq t\leq t_i$ in each iteration $i$. To simulate that, we then add $1000$ new records into the dataset in each iteration. To establish a comparison with Laplace and Pain-Free mechanisms, we implemented those two mechanisms in an interactive manner between the  MSSP and data owner. However, DPOAD would sample the sensitivity in each iteration according to the MSSP-updated PDF, while Laplace and Pain-Free utilize the global sensitivity and uniformly sampled sensitivity, respectively.
 
As shown in Table~\ref{tab:rounds}, for all mechanisms, we first observe that both precision and recall increase with more analysis iterations by the MSSP. However, as the number of iterations increases, the precision and recall of Laplace and Pain-free mechanisms only have minor improvement compared to DPOAD. Thus, we also observe that the DPOAD can achieve a much better accuracy on all the datasets (e.g., $\sim$500\% better recall results than Laplace and Pain-Free on the credit card clients dataset).
 

The main reason is that DPOAD benefits from the MSSP-updated PDF to select the sensitivity in each iteration, the sensitivity should be smaller as the iteration number increases. 
On the contrary, since the sensitivity in Laplace and Pain-Free has to be very large in each iteration, the accuracy of these two benchmarks cannot be high. 

 \subhead{Summary} 
 The above experiments have validated the effectiveness of DPOAD to improve the precision and recall in anomaly detection while ensuring $(\epsilon,\gamma)$-differential privacy. A better accuracy could also be achieved by performing the analysis with more iterations, e.g., 6 analysis iterations in our experiments could achieve as high as 0.84 precision and 0.67 recall by DPOAD (the results can be even better if we relax the privacy protection). All the experiments on different datasets have validated the practicality of DPOAD.

\section{Related Work}
\label{relat}
 
Most existing works recognized the paradox between anomaly detection and differential privacy (DP) due to the indistinguishability property. Okada et al. \cite{DBLP:conf/pkdd/OkadaFS15} first highlights such a conflict. To solve this paradox, a DP relaxation is inevitable, particularly for existence-dependent queries, to ensure a certain level of detection accuracy. Thus, all DP-based works have defined their own relaxed DP notion, which generally make them applicable to limited cases (e.g., Anomaly-restricted privacy \cite{DBLP:conf/cscml/BittnerSW18}, and sensitive privacy \cite{DBLP:conf/ccs/AsifPV19}). In contrast, we define our relaxed DP with the Random Differential Privacy (RDP) \cite{hall2013random} to protect specific percent of data with $\epsilon$-DP.



Furthermore, most of the existing works do not consider outsourcing the anomaly detection task as they mainly rely on the fact that the anomaly detection is locally performed on the data and only the analysis result is provided with DP. Instead, DPOAD performs anomaly detection on the outsourced data (generated by the DP algorithm), which is more practical since the clients do not need to perform expensive analysis.

Okada et al. \cite{DBLP:conf/pkdd/OkadaFS15} introduce a mechanism based on the smooth upper bound of the local sensitivity (a relaxed DP) and use it in a limited setting. The approach does not report the detected outliers but cover two types of DP queries: outlier counting and top-h subspaces with a large number of outliers. Similarly, several other approaches (e.g., \cite{DBLP:conf/cscml/BittnerSW18,DBLP:conf/coinco/AsifTVSA16,DBLP:journals/ijcis/AsifTVSA18,bohler2017privacy}) propose to relax/generalize DP and design anomaly detection methods with DP. Most of such approaches are data-dependent and can only be performed under a specific setting. 
For instance, Bittner et al. \cite{DBLP:conf/cscml/BittnerSW18} defines the anomaly-restricted DP and proposes a group-based search algorithm for this relaxed notion. However, their approach relies on input datasets that have a single outlier, which is different from our approach. Asif et al. \cite{DBLP:conf/coinco/AsifTVSA16,DBLP:journals/ijcis/AsifTVSA18} considers DP in the collaborative outlier detection where data is either vertically or horizontally distributed among multiple parties. However, such approach works only for datasets with certain characteristics (i.e., categorical data) and its extension to numeric data using discretization can limit the types of detectable outliers with reduced data usability. In contrast to them, we consider a more general and practical model including any number of anomalies.

Moreover, Asif et al. \cite{DBLP:conf/ccs/AsifPV19} generalizes DP and defines sensitive privacy, which determines sensitive records after quantifying the database.  
B{\"o}hler et al. \cite{bohler2017privacy} assumes that the data owner already knows the outliers and excludes them from the dataset before adapting the sensitivity to the remaining data. It also relies on the trust between the analyst and a new entity called the correction server whose role is to increase the accuracy of the outlier detection. 

\vspace{-0.05in}
\section{Conclusion}
\label{sec:conclusion}
We propose a novel solution called DPOAD to address the problem of privacy-preserving outsourcing of anomaly detection by decreasing the privacy of the anomalies through communication between the data owners and the data analysts. We show that DPOAD can significantly improve the accuracy of the analysis for the data with abnormal behaviour while preserving the privacy of the data with normal behavior. The system receives the input from the data owners and the data analysts. Each data owner, to build a reliable estimation of the sensitivity, releases sufficient number of DP version of their dataset (non suitable for anomaly detection) to analysts. The analyst provides owners with distribution of their data which can be used to estimate an updated sensitivity value using the calculated anomaly scores for records. The framework can be used in settings with one or more data owners. We formally benchmark DPOAD under DP-Histogram publishing and showcase the performance improvements.

\bibliographystyle{unsrt}
\bibliography{acmart}



\end{document}